\documentclass[prl,reprint,aps,superscriptaddress,showpacs,nofootinbib]{revtex4-2}

\usepackage[utf8]{inputenc}
\usepackage[colorlinks=true,linkcolor=blue,citecolor=blue,urlcolor=blue]{hyperref}
\usepackage{setspace}
\usepackage{graphicx}
\usepackage{amsmath,amssymb}
\usepackage{dsfont}
\usepackage{xcolor}
\usepackage{amsthm}
\usepackage{float}
\usepackage{dsfont}
\usepackage{braket}
\usepackage{algorithm}
\usepackage{algpseudocode}
\usepackage{stmaryrd}
\usepackage[caption=false]{subfig}

\algnewcommand\algorithmicforeach{\textbf{for each}}
\algdef{S}[FOR]{ForEach}[1]{\algorithmicforeach\ #1\ \algorithmicdo}

\newtheorem{lemma}{Lemma}
\newtheorem{definition}{Definition}
\newtheorem{theorem}{Theorem}
\newtheorem{corollary}{Corollary}
\theoremstyle{remark}
\newtheorem{remark}{Remark}

\definecolor{FGreen}{RGB}{1,68,33}

\makeatletter
\def\l@subsection#1#2{}
\def\l@subsubsection#1#2{}
\makeatother

\begin{document}

\title{Low-overhead fault-tolerant quantum computation by gauging logical operators}
\author{Dominic J.~Williamson}
\affiliation{IBM Quantum, IBM Almaden Research Center, San Jose, CA 95120, USA}
\author{Theodore J.~Yoder}
\affiliation{IBM Quantum, IBM T.J. Watson Research Center, Yorktown Heights, NY 10598, USA}
\date{July 7, 2026}

\begin{abstract}
\noindent
Quantum computation must be performed in a fault-tolerant manner to be useful in practice. Recent progress has established quantum error-correcting codes with sparse connectivity requirements and constant qubit overhead suitable for quantum memory. However, existing schemes that include fault-tolerant logical measurement on such quantum memories do not always achieve low qubit overhead. Here we present a low-overhead method to implement fault-tolerant logical measurement on a quantum error-correcting code by treating the logical operator as a physical symmetry and gauging it so that it is enforced by a product of local symmetries. The gauging measurement procedure introduces a high degree of flexibility that can be exploited to achieve a qubit overhead that is linear in the weight of the operator being measured up to a polylogarithmic factor. This flexibility also allows the procedure to be adapted to arbitrary quantum codes. Our results provide a more efficient approach to performing fault-tolerant quantum computation, making it more tractable for near-term implementation. 
\end{abstract}

\maketitle

Quantum error-correcting codes are an essential ingredient to protect quantum information in a quantum computer from errors due to coupling to an external environment~\cite{Shor1995,Steane1996,shor1996fault,gottesman1997stabilizer,aharonov1997fault,knill1998resilientQC,kitaev1997quantum,preskill1998reliable,preskill1997fault}. 
Recent progress~\cite{gottesman2014fault,tillich2014quantum,leverrier2015quantum,Panteleev2019,evra2020decodable,hastings2020fiber,panteleev2020quantum,breuckmann2020balanced,breuckmann2021ldpc} has led to the discovery in Ref.~\cite{panteleev2022asymptotically} of \textit{good} quantum low-density parity-check (qLDPC) codes~\cite{panteleev2022asymptotically,leverrier2022quantum,Dinur2023}, which have constant encoding rate and relative distance. 
Such codes are far more efficient at protecting large amounts of quantum information~\cite{panteleev2021degenerate,bravyi2024high,Xu2023} than standard approaches based on the surface code~\cite{kitaev2003fault,bravyi1998quantum,dennis2002topological}. 
An important question is whether the advantages that good qLDPC codes offer for quantum information storage come at the cost of straightforward and efficient logical quantum information processing. 

A key requirement for logical quantum gates on a quantum error-correcting code is fault tolerance -- they must function in the presence of errors while maintaining protection of the encoded quantum information throughout a computation. 
Existing approaches to performing fault-tolerant logical quantum gates can be roughly divided into code-preserving and code-deforming. 

Code-preserving gates include transversal gates~\cite{gottesman1997stabilizer}, which act via the same operator on all physical and logical qubits, and more general locality-preserving gates~\cite{doi:10.1063/1.4939783}. 
Nontrivial code-preserving gates are synonymous with symmetries of an underlying code~\cite{Bombin2006TQD,yoshida2015gapped,doi:10.1063/1.4939783,moussa2016transversal,Quintavalle2022,Webster2023,breuckmann2024fold}. 
As a consequence, such gates necessitate a degree of structure to exist in the code. 

Code-deforming gates are more general and allow a code to be deformed through a sequence of different codes to enact a logical gate~\cite{Raussendorf2007,horsman2012surface,Landahl2014,Paetznick2013,TJOC2013,Bombin2015,Anderson2014}. 
Unlike code-preserving gates, code-deforming gates are generic as they do not necessitate the existence of nontrivial symmetries of the code beyond the logical operators themselves. 
In particular, code-deforming gates allow for the fault-tolerant measurement of logical operators~\cite{horsman2012surface,Landahl2014}. 
This opens up the possibility of implementing measurement-only quantum computation, where Clifford gates and $T$-gate injections are implemented via measurement, or 
Pauli based computation (PBC)~\cite{bravyi2016trading} where a quantum circuit is compiled into a sequence of generalized magic state injections~\cite{litinski2019game}. 
Lattice surgery is a code-deforming approach to fault-tolerant computation for the surface code that has been studied extensively~\cite{horsman2012surface,Fowler2018,litinski2019game,Chamberland2021}.

There are a number of proposals for performing fault-tolerant logical gates on qLDPC codes in the literature~\cite{Zhu2023,Scruby2024,cohen2022low,Cowtan2023,Cowtan2024,cross2024linear,Zhang2024}. 
These include code-preserving gates and code-deforming gates. 
The most direct generalization of surface code lattice surgery is the scheme proposed in Ref.~\cite{cohen2022low}, and refined in Ref.~\cite{cross2024linear}, which applies to general qLDPC codes and makes use of an auxiliary system similar to a patch of surface code. 
The auxiliary system itself incurs a qubit overhead of $\Theta(Wd)$ where $W$ is the weight of the logical being measured and $d$ is the code distance. 
This auxiliary overhead can make fault-tolerant measurement-based computation significantly more expensive than implementing a quantum memory. 
For instance, a good quantum code that uses $n=\Theta(d)$ physical qubits to encode $k=\Theta(d)$ qubits (see Refs.~\cite{panteleev2022asymptotically,leverrier2022quantum} for examples) requires an auxiliary system at least $\Omega(n^2)$ in size, significantly larger than the code itself. 
Even for other constant rate $k=\Theta(n)$ codes with polynomial distances $d=\text{poly}(n)$, such as hypergraph product codes~\cite{tillich2014quantum,leverrier2015quantum}, measuring generic logical operators with weight $W=\Theta(n)$, such as in PBC, uses a super-linear number $n^{\omega(1)}$ of qubits in the auxiliary system.



In this work we introduce a general procedure for the fault-tolerant measurement of a logical operator in a stabilizer code by treating it as physical symmetry and gauging it via measurement~\cite{Williamson2020a,Tantivasadakarn2021,Tantivasadakarn2022}. 
The gauging procedure is common in the theory of condensed matter and high energy physics, it enforces a global symmetry via a product of local symmetries. 
Our measurement scheme is based on the fact that the gauging transformation makes it possible to infer the measurement of a logical operator via a product of local stabilizers in a deformed code. 
Gauging has been employed widely in theoretical physics to construct new models and establish relationships between known models~\cite{Kramers1941,Wegner1971,kogut1975hamiltonian}. 
The connection of gauging to quantum codes has also been studied~\cite{Williamson2016,Vijay2016,kubica2018ungauging,Dolev2021,Rakovszky2023}. 
Previous work focused on gauging procedures that correspond to the simultaneous initialization and readout of all logical qubits in a code and did not consider fault tolerance. 
The appearance of a similar procedure in recent work on weight reduction for quantum codes hints at wider applications of gauging in quantum error correction~\cite{Bacon2015,hastings2016weight,hastings2021quantum,wills2023tradeoff,sabo2024weight,WireCodesPreprint}. 

Here, we go beyond previous work by developing a fault-tolerant gauging measurement procedure that can precisely address arbitrary individual logical operators in a large code block. 
We prove that the worst-case qubit overhead of the gauging measurement procedure applied to an arbitrary Pauli operator of weight $W$ is $O(W \log^3 W)$, a significant improvement over existing results in the literature. 
We further demonstrate that the flexibility inherent to the gauging measurement procedure leads to better performance than existing schemes for logical measurement~\cite{Cowtan2024,cross2024linear} even for small instances of bivariate bicycle (BB) codes~\cite{bravyi2024high}, see the BB code example in the supplementary information~\cite{supplement}. \newline

\begin{figure*}[t]

\hspace{15pt}
\subfloat[Gauging measurement]{\raisebox{3pt}%
{\includegraphics[width=0.38\textwidth]{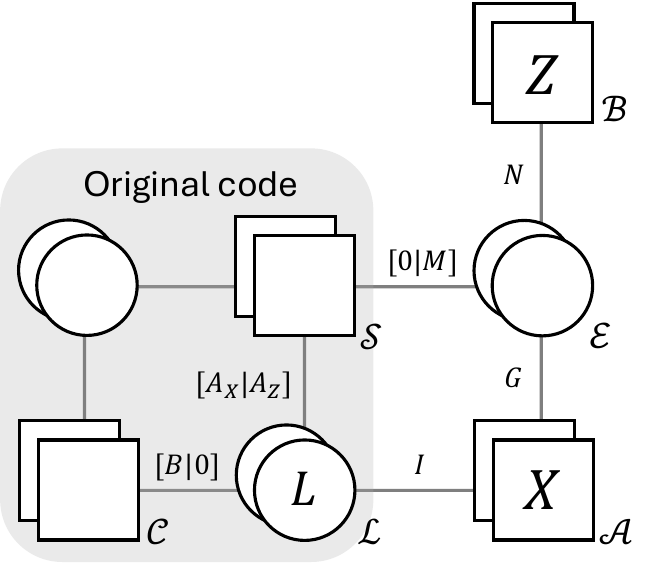}}
\label{fig:thin_Tanner}
}\hspace*{\fill}
\subfloat[Surface code lattice surgery]{\raisebox{28pt}%
{\includegraphics[width=0.48\textwidth]{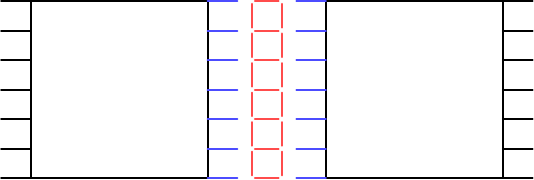}}
\label{fig:SurfaceCodes}
}\hspace{15pt}

\caption{\textbf{Tanner graph of the Gauging measurement and a Surface code example:}  \textbf{(a)} The Tanner graph of the deformed code can be represented compactly by creating ordered subsets of qubits (circles) and checks (squares) of $X$, $Z$, or mixed (unlabeled) type. Each edge connecting an $X$ or $Z$ check set $\mathcal{P}$ and qubit set $\mathcal{Q}$ is labeled by a binary matrix $H$ with $|\mathcal{P}|$ rows and $|\mathcal{Q}|$ columns, where $H_{ij}=1$ if and only if the $i^\text{th}$ check in $\mathcal{P}$ acts non-trivially on the $j^\text{th}$ qubit in $\mathcal{Q}$. Edges from mixed type check sets are instead labeled with a symplectic matrix of the form $[H_X|H_Z]$ where $H_X$ indicates qubits acted on by $X$ and $H_Z$ those acted on by $Z$. The original code may not be CSS, but we assume without loss of generality that $L$, the operator being measured, is $X$-type and its qubit support is $\mathcal{L}$. Set $\mathcal{C}$ contains checks from the original code that do not have $Z$-type support on $\mathcal{L}$, while set $\mathcal{S}$ contains checks that do. Also, $\mathcal{A}$ is the set of Gauss's law operators $A_v$, $\mathcal{B}$ is the set of flux operators $B_p$, $\mathcal{E}$ is the set of edge qubits, and we abuse notation so that $G$ also denotes the auxiliary graph's incidence matrix. Matrix $N$ specifies a cycle basis of $G$ and $M$ indicates how original stabilizers are deformed by perfect matching in $G$. \textbf{(b)} Applying the gauging measurement procedure to a product of $X$-type logicals on a pair of surface codes (blue edges) and choosing the graph $G$ to be a ladder (red edges) results in a standard surface code lattice surgery procedure.}
\end{figure*}

\noindent\textbf{\large Results}

\noindent\textbf{Gauging measurement} ---
We now present the procedure to measure a logical operator via gauging. 
We focus on the task of measuring a logical representative $L$ in an $[[n,k,d]]$ qLDPC stabilizer code on qubits specified by checks $\{s_i\}$~\cite{gottesman1997stabilizer}. 
Choosing an appropriate basis for each qubit, we ensure that $L$ is a product of Pauli-$X$ matrices without loss of generality. 
We view $L$ as a symmetry of the code by identifying the codespace with the groundspace of the Hamiltonan $H=-\sum_i s_i$ and noting that $L$ commutes with $H$. 
From this point of view, it is natural to project into an eigenspace of $L$ by gauging it. 

To gauge $L$ we first pick a connected graph $G$ whose vertices $V$ are identified with the qubits in the support of $L$.
We sometimes refer to this graph as the auxiliary graph.
Next, we introduce auxiliary qubits on the edges $E$ of $G$, one per edge, initialized in the $\ket{0}$ state. 
In the physics literature these edge qubits are referred to as \textit{gauge qubits}. 
Employing these edge qubits, we introduce a set of \textit{Gauss's law operators}, $\mathcal{A}=\{A_v\}_{v\in V}$ where ${A_v=X_v\prod_{e\ni v}X_e}$, and crucially $L=\prod_vA_v$, so that we can infer the eigenvalue of $L$ from the eigenvalues of all the $A_v$. To obtain these eigenvalues, we measure each of the $A_v$ operators. 
This results in a deformed code, which is a subsystem code if the number of edge qubits exceeds the number of Gauss's law operators. 
However, this subsystem code is gauge fixed by the choice of initial state for the edge qubits. 
Specifically, include in the stabilizer group a set of \textit{flux operators $\mathcal{B}=\{B_p\}_{p\in C}$, where $B_p=\prod_{e\in p}Z_e$, and $p$ labels a generating set of cycles for $G$.
The flux operators commute with the $A_v$ operators by construction and can be obtained by deforming the initial single-qubit $Z_e$ operators on edge qubits following standard stabilizer update rules~\cite{gottesman1997stabilizer}.}   

The measurement of operators in $\mathcal{A}$ also necessarily deforms the check operators of the original code. Specifically, those check operators $s_i$ that have $Z$-type support (i.e.~the set of qubits acted on by a $Y$ or $Z$) intersecting with the support of $L$ do not commute with operators in $\mathcal{A}$.
Thus, in the deformed code $s_i$ is replaced by $s_i$ times some set of single-qubit $Z_e$ operators on edge qubits. 
This set of edge qubits has a conceptually simple form in relation to the graph $G$. 
Because $L$ is logical in the original code, $s_i$ commutes with it, and so the size of the $Z$-type support of $s_i$ on $L$ is even.
The corresponding even-sized set of vertices in $G$ can be paired up via a perfect matching $\mu_i$, i.e.~a set of edges whose $\mathbb{Z}_2$ boundary is equal to the set of vertices. Then $\tilde s_i = s_i\prod_{e\in\mu_i}Z_e$ is an appropriately deformed version of $s_i$ which commutes with $\mathcal{A}$. Because the flux operators $\mathcal{B}$ defined on cycles of the graph have $+1$ eigenvalues, different choices of perfect matching result in the same code space.

See Fig.~\ref{fig:thin_Tanner} for a depiction of the Tanner graph of the deformed code including the auxiliary edge qubits, the Gauss's law operators $\mathcal{A}$, the flux operators $\mathcal{B}$, and the deformed checks $\{\tilde s_i\}$.

The above gauging procedure maps the original code space into the code space of a \textit{deformed code} (also known as a gauged code in physics literature~\cite{Dolev2021}). 
The gauged code supports fully mobile charges that are created and moved around via strings of $Z_e$ operators, and fluxes that are created via $X_e$ operators. 
These operators can be viewed as generalized higher-form symmetries, and share a mutual braiding anomaly due to the commutation relations of the Pauli matrices. 
To implement the eigenspace projections of the gauging procedure we rely on measuring the Gauss's law operators~\cite{Williamson2020a,Tantivasadakarn2021}. 
While the initial measurement of each Gauss's law operator may return a random eigenvalue outcome, the product of all outcomes yields the measured eigenvalue of the logical $L$. 
We can return to the original code space by simply projecting onto the simultaneous $+1$ eigenspace of the $Z_e$ operators on all the edge qubits. 
To implement this step we measure $Z_e$ on each edge qubit. 
While the outcome of each measurement may be random, it is possible to find a Pauli-$X$ \textit{byproduct operator} on the original qubits that can be applied to achieve the same effect as projecting onto the all $+1$ measurement outcome.
This step is commonly referred to as ungauging~\cite{Williamson2016}. 
These steps are collected in Algorithm~\ref{alg:GaugeLogical}.

\begin{figure}[t]
\begin{algorithm}[H]
\caption{Gauging measurement procedure}
\label{alg:GaugeLogical}
\begin{algorithmic}
    \Require A Pauli stabilizer code specified by checks $\{s_i\}$ initialized in the physical code state $\ket{\psi}$. \\
    A Pauli logical operator representative $L$. \\
    A connected graph $G=(V,E)$ with an arbitrarily chosen vertex $v_0$ and an isomorphism that identifies the qubits in the support of $L$ with the vertices of $G$. 
    
    \Ensure The result $\sigma=\pm1$ of measuring $L$ and the post-measurement code state $\ket{\Psi}=\frac{1}{2}(\mathds{1}+\sigma L)\ket{\psi}$. 

    \State $\sigma \gets 1$
    \State $\ket{\Psi} \gets \ket{\psi}$
    \State $\{\omega_e\}_{e\in E} \gets \{1\}_{e\in E}$
    \State $\gamma \gets \{\}$
    
    \ForEach{edge $e$ in $G$} 
    \State $\ket{\Psi} \gets \ket{\Psi}\otimes \ket{0}_e$ 
    \Comment{Initialize a qubit on $e$}
    \EndFor
    \ForEach{vertex $v$ in $G$}
	\State Measure $A_v$ on $\ket{\Psi}$
        \Comment{$A_v=X_v\prod_{e\ni v}X_e$}
        \State $\varepsilon \gets $ Measurement result 
        \Comment{Measurement result is $\pm 1$}
        \State $\ket{\Psi} \gets \frac{1}{2}(\mathds{1}+\varepsilon A_v)\ket{\Psi}$
        \Comment{Post-measurement state}
        \State $\sigma \leftarrow \varepsilon \sigma$
    \EndFor
    \ForEach{edge $e$ in $G$}
	\State Measure $Z_e$ on $\ket{\Psi}$
        \State $\omega _e \gets $ Measurement result
        \Comment{Measurement result is $\pm 1$}
        \State $\ket{\Psi} \gets \frac{1}{2}(\mathds{1}+\omega_e Z_e)\ket{\Psi}$
        \Comment{Post-measurement state}
        \State Discard auxiliary qubit on $e$
    \EndFor
    \ForEach{vertex $v$ in $G$}
    \State $\gamma \gets$ An arbitrary edge-path from $v_0$ to $v$
        \If{$\prod_{e\in \gamma}\omega_e=-1$}
            \State $\ket{\Psi} \gets X_v \ket{\Psi}$
            \Comment{Apply byproduct operator}
        \EndIf
    \EndFor
\end{algorithmic}
\end{algorithm}
\end{figure}

\begin{theorem}[Gauging measurement]
\label{thm:GaugingMeasurement}
    The gauging procedure defined in Algorithm~\ref{alg:GaugeLogical} is equivalent to performing a projective measurement of $L$.
\end{theorem}
\begin{proof}
The proof is contained in the Methods section. 
\end{proof}

We note that the auxiliary graph $G$ can be chosen to have additional vertices beyond the qubits in the support of a given logical operator. 
Conceptually, this is achieved by adding an extra qubit, initialized in the $\ket{+}$ state, for each desired extra vertex and then gauging $L$ multiplied by $X$ on each extra vertex, instead of $L$ itself.
We refer to extra vertices as \textit{dummy vertices}.
In practice, the qubits on dummy vertices do not need to exist -- as they are initialized in the $\ket{+}$ state, their contribution to the Gauss's law operators is just the deterministic value $+1$.

\noindent\textbf{Graph desiderata and constructions} ---Algorithm~\ref{alg:GaugeLogical} is an extremely flexible recipe for logical measurement because the choice of a connected auxiliary graph $G$ is arbitrary. 
However, the properties of the deformed code depend strongly on the choice of this graph. 
For instance, the gauging measurement procedure has qubit overhead equal to the number of edges in the graph $G$ (and involves an additional, proportional number of checks).
Next, we list desirable properties of $G$ to achieve a fault-tolerant implementation of the gauging measurement with low overhead.

\begin{definition}[Suitable Graph]\label{defn:GDesiderata}
Given a logical operator $L$ in a qLDPC code, a suitable graph $G$ for the gauging measurement procedure is defined to be a connected graph that results in a deformed code that is LDPC and has code distance at least that of the original code. 
\end{definition}

\begin{theorem}[Graph Desiderata]\label{thm:GDesiderata}
Consider the gauging measurement procedure for a logical operator $L$ in a qLDPC code with a choice of connected graph $G$. 
If the following desiderata are satisfied, $G$ is a suitable graph.
\begin{enumerate}
\item The graph $G$ is sparse, i.e.~has constant degree. 
\item The matchings $\mu_i$ for all checks $s_i$ are of constant size, and any edge qubit is in no more than a constant number of these matchings.
\item There is a generating set of constant-weight cycles for $G$ such that any edge is in a constant number of cycles from the generating set.
\item The graph $G$ is sufficiently expanding in the sense that its Cheeger constant $h(G)$ is at least 1.
\end{enumerate}
\end{theorem}

\begin{proof}
The first three desiderata are shown to hold if and only if the deformed code is LDPC by explicitly defining a generating set of checks for the deformed code, Lemma~\ref{lem:deformed} in Methods, and inspecting that generating set for the LDPC property. Lemma~\ref{lem:spacedistance} in Methods establishes that the expansion property in desideratum 4  implies a sufficient deformed code distance.
\end{proof}

\begin{theorem}[Construction of a suitable $G$]\label{thm:WorstCaseG}
    Consider an arbitrary weight-$W$ logical $L$ in a qLDPC code, a suitable graph $G$ can be constructed with qubit overhead $O(W \log^3 W)$. 
\end{theorem}

\begin{proof}
    We now outline a construction that produces a graph $G$ with qubit overhead $O(W \log^3 W)$ that satisfies the desiderata in Theorem~\ref{thm:GDesiderata}. 
    The graph starts with $W$ vertices associated to the qubits in the support of $L$.
    First, to satisfy desideratum 2, for each check $s_i$ we pair up the qubits in the intersection of the $Z$-type support of $s_i$ with the support of $L$, and add an edge to $G$ for each pair of matched qubits. 
    This step results in a constant degree graph due to the LDPC property of the input code. 
    Second, to satisfy desideratum 4, we add edges to $G$ until $h(G)\geq 1$. 
    This step can be performed by adding edges to $G$ at random while preserving constant degree, or by taking an existing constant degree expander graph with $h(G)\geq 1$ and adding its edges to $G$. 
    Both of the above methods lead to a constant degree graph. 
    Third, we take the cartesian graph product \cite{sabidussi1959graph,vizing1963cartesian} of $G$ with a path graph on $r$ vertices. In other words, this step takes $r$ copies of $G$ and, for all $i=1,2,\dots,r-1$, connects copy $i$ to copy $i+1$ by adding an edge between each vertex in copy $i$ and the corresponding vertex in copy $i+1$.
    Fourth, add edges to the additional layers to sparsify the cycle basis of this product graph to achieve a desired constant weight cycle basis. 
    This results in a final graph that satisfies the first three desiderata by construction, see Methods Lemma~\ref{lem:SparsifiedDesiderata}. 
    The last two steps are based on the Freedman-Hastings decongestion lemma \cite{freedman2021building} and cellulation, respectively, and their purpose is to satisfy desideratum 3 by creating a graph with a sparse cycle basis without compromising the other desiderata that were established in steps one and two. 
    Desideratum 4 may not be satisfied by the larger graph created during the decongestion step, however we show in the Methods, Corollary~\ref{cor:SparsifiedDistance}, that so long as the original graph $G$ satisfies desideratum 4 the deformed code distance is at least that of the original code. 
    The Freedman-Hastings decongestion lemma establishes a $\log^3 W$ upper bound on the number of layers needed to sparsify the cycles of an arbitrary constant degree graph $G$ with $W$ vertices when following the third and fourth steps; see Ref.~\cite{he2025extractors} for a detailed review of this point.
    The lemma comes with an efficient algorithm to construct a sparsified basis of cycles, given a constant degree graph. See Fig.~\ref{fig:thick_Tanner} for a depiction of the deformed code after these steps. 
\end{proof}

\begin{figure*}[t]
\includegraphics[width=0.9\textwidth]{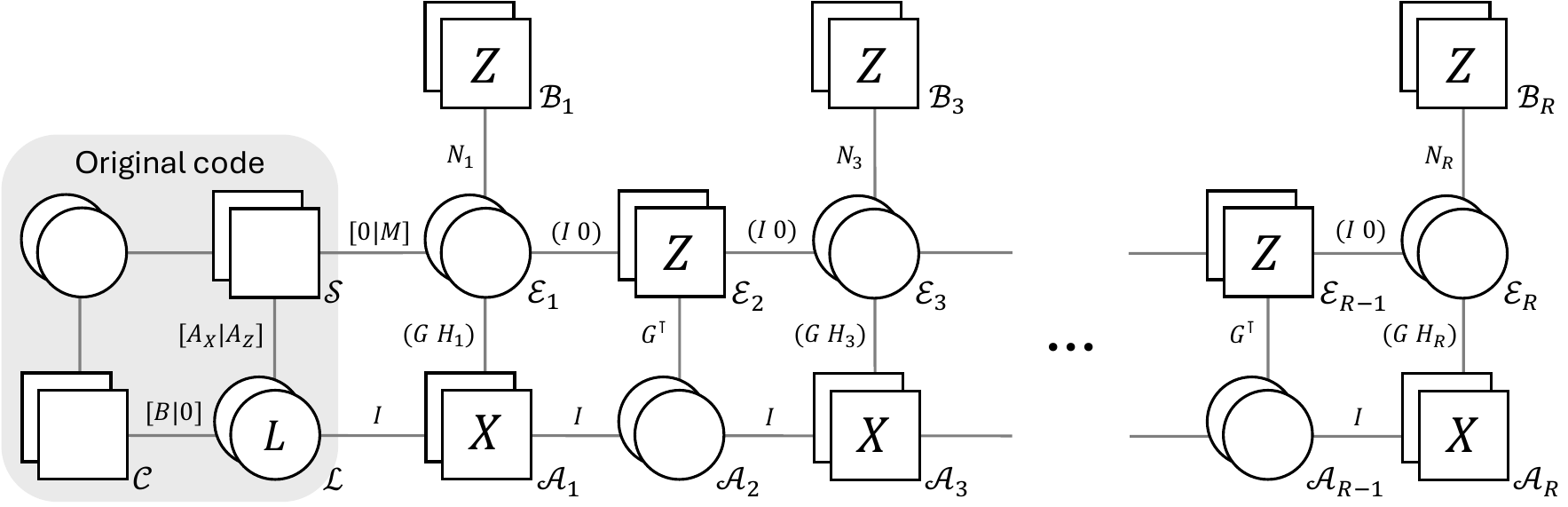}
\caption{\label{fig:thick_Tanner} \textbf{Tanner graph of the decongested and cellulated gauging measurement:} The Tanner graph of the complete construction including decongestion and cellulation to guarantee the deformed code is LDPC, see Theorem~\ref{thm:WorstCaseG}. We use the same notation as explained in Fig.~\ref{fig:thin_Tanner}. Here, as a result of the cartesian product operation in step three of Theorem~\ref{thm:WorstCaseG}, the $r=(R+1)/2$ copies of a base expanding graph $G$ have vertices and edges $(\mathcal{A}_j,\mathcal{E}_j)$ for odd $j$. Notation $(B_1\text{\space}B_2)$ indicates a block matrix constructed from submatrices $B_1$ and $B_2$. We use block matrices where qubits of the left block are those on edges of the expanding graph $G$ and qubits of the right block are those added to cellulate cycles in step four of Theorem~\ref{thm:WorstCaseG}. Depending on the layer, cellulation is done to different elements of the cycle basis of $G$. Using $r=O(\log^3W)$ layers ensures the deformed code is LDPC by the decongestion lemma \cite{freedman2021building}.}
\end{figure*}

\noindent\textbf{Fault-tolerant implementation} ---
Even if measurements are faulty, algorithm~\ref{alg:GaugeLogical} can be implemented fault tolerantly by measuring the stabilizer checks of the original code for $d$ rounds, followed by measuring the checks of the deformed code for $d$ rounds, and finally measuring the checks of the original code for a further $d$ rounds. 

\begin{theorem}[Fault tolerance]
\label{thm:fault-tolerance}
    The fault-tolerant implementation of Algorithm 1 with a suitable graph has spacetime fault-distance $d$.
\end{theorem}

\begin{proof}
    The proof follows several lemmas that are stated and proved in the Fault-tolerant implementation section of the methods and the Spacetime code and spacetime fault-distance section of the supplementary information~\cite{supplement}. 
    The argument is structured as follows, the fault-distance is shown to be lower bounded by the minimum of the space fault-distance and the time fault-distance. 
    The time fault-distance is lower bounded by the number of rounds between the start and end of the code deformation, which is chosen to be $d$. 
    The space fault-distance is lower bounded by the distance of the original code multiplied by $\min(1,h(G))$.
    This is because any logical in the deformed code can be cleaned such that it defines a logical of the original code. 
    The distance reduction of the cleaning is determined by the connectivity of $G$.
\end{proof}

The $d$ rounds of quantum error correction in the original code before and after the gauging measurement are for the purposes of establishing a proof and may be overkill in practice. 
In a full fault-tolerant computation the number of rounds required before and after a gauging measurement depends on the surrounding operations. 
If the gauging measurement occurs in the middle of a large computation, it could be that even a constant number of rounds before and after are sufficient to ensure fault tolerance. 
It is an open question to fully characterize when fault tolerance is achieved with fewer than $d$ rounds between gauging measurements, and how this choice affects the threshold depending on the surrounding operations.

\noindent\textbf{\large Discussion}

In this work we have introduced a method to implement low-overhead fault-tolerant quantum computation with high-distance qLDPC codes. 
Our method is based on viewing a logical operator as a symmetry and gauging it via local measurements. 
The gauging measurement procedure allows a high degree of flexibility, making it a promising method that can be optimized for even small code instances with near-term applications. 
We have taken advantage of this flexibility to devise specific gauging measurements for examples of bivariate bicycle codes that are more efficient than all existing measurement schemes in the literature~\cite{cross2024linear,cohen2022low}.

Gauging measurement goes beyond previous approaches to qLDPC code surgery, which was initiated in Ref.~\cite{cohen2022low} and developed in subsequent works \cite{Cowtan2024,cross2024linear,Zhang2024,Xu2024Fast}. 
Our key innovation is the introduction of a flexible ancilla system that is not wholly defined by the structure of the logical in the original code. 
This allowed us to find a fault-tolerant qLDPC code surgery procedure to measure arbitrary Pauli logicals with a worst-case scaling that is linear in the weight of the target logical, up to polylogarithmic factors. 
It also allowed us to find more efficient logical measurements on instances with small block size. 
The flexibility of the gauging measurement procedure can be used to recover a number of well-known existing schemes for logical measurement including surface code lattice surgery~\cite{horsman2012surface}, the measurement scheme in Ref.~\cite{cohen2022low} and the modified version in Ref.~\cite{cross2024linear}. See the section on recovering existing protocols from the gauging measurement framework in the supplementary information~\cite{supplement} for more discussion about how these and other existing protocols are special cases of gauging measurement.

Some aspects of our approach to designing high-performance gauging measurements extend previous work on qLDPC code surgery. 
Ref.~\cite{cross2024linear} in particular points out the importance of sufficient hypergraph expansion to ensure a large fault-distance. 
As a result of our work, we are able to guarantee appropriate expansion in the gauging measurement of an arbitrary logical operator by adding edges to the auxiliary graph. 
The gauging measurement also implements gauge-fixing of the deformed code, as in Ref.~\cite{cross2024linear}, by measuring the checks associated to cycles in the auxiliary graph. 
Importantly, the graph based approach inroduced in the gauging measurement allows us to construct a sparse cycle check basis for any code by employing cellulation and decongestion \cite{hastings2021quantum,freedman2021building} to ensure the deformed code remains LDPC. This goes beyond previous ad hoc approaches to gauge fixing on specific instances without a worst-case guarantee as in Ref.~\cite{cross2024linear}. 

Concurrent with our work, Ref.~\cite{XanaduDraft} introduced a similar approach to qLDPC code surgery that involves gauge fixing with cellulation, and boosting expansion by adding edges to an auxiliary graph. There, the focus was on logicals of a fixed Pauli type in CSS codes, while our focus is on general Pauli logicals in stabilizer codes which can be non-CSS. Furthermore,  Ref.~\cite{XanaduDraft} does not consider decongestion, and as such does not provide a worst-case guarantee on the scaling of their logical measurement procedure. 
In a follow up work, Ref.~\cite{Swaroop2024} builds directly on the gauging measurement framework by introducing a graph construction algorithm that is used to find an efficient procedure to perform joint logical measurement of disjoint logical operators on one or more qLDPC code blocks. This operation is a gauging measurement with a different choice of auxiliary graph to the worst-case construction outlined above. 
In another follow up work, Ref.~\cite{yuan2026parsimoniousquantumlowdensityparitycheck} introduced an alternative to the Freedman-Hastings decongestion lemma to construct a generic suitable graph with worst-case scaling $O(W \log W)$. 

The progress reported in this work raises a number of directions that invite further investigation. 
For what families of good qLDPC codes can the gauging measurement be performed with linear qubit overhead? 
Can the generalization of the gauging measurement to a hypergraph be used to measure a large number of commuting but overlapping logical operators simultaneously?
Can we introduce meta-checks to the generalized gauging measurement procedure to make the measurement single-shot with a constant time overhead?
How should the fault-tolerant gauging measurement be decoded? 
We expect that a general purpose decoder based on belief propagation with ordered statistics post-processing can be used, however it should be possible to take advantage of the extra structure inherent to the procedure to find a decoder with better performance. 
Such a decoder could incorporate matching on the $A_v$ syndromes similar to the approach in Ref.~\cite{cross2024linear}. 

The directions listed here can be interpreted as a list of desiderata for an ideal fault-tolerant logical measurement procedure: it should implement PBC, be fully parallelized, and have constant relative overhead in space and time. 
It remains an open question whether any fault-tolerant logical measurement can satisfy the above desiderata, or how close one can get. 
Even the answer to the classical analog of this question is not known.

\section*{acknowledgements}

We thank Tomas Jochym-O'Connor and Guanyu Zhu for useful discussions, Alexander Cowtan for useful comments on a draft, Esha Swaroop for bringing our attention to the Hastings-Freedman decongestion lemma, and Sunny He for explaining how the $\log^3$ factor results implicitly from that lemma's proof. 
TY thanks Andrew Cross for sharing his method of using CPLEX to calculate code distances. 
DW thanks Lawrence Cohen for an inspiring discussion on the ferry back from Rottnest Island. 
This work was initiated while DW was visiting the Simons Institute for the Theory of Computing.

\section*{Author contributions}
DJW and TJY both contributed extensively to the paper.

\section*{Competing interests}
US Patent Application 18/897427 (filed on 26 September 2024 and naming DJW and TJY as co-inventors) contains technical aspects from this paper.

\bibliographystyle{apsrev4-1} 
\bibliography{references}

\vspace{1cm}

\noindent\textbf{\large Methods}

\noindent\textbf{Gauging measurement} --- Here we prove Theorem~\ref{thm:GaugingMeasurement}, showing that gauging measurement correctly performs a logical measurement. We first introduce definitions of boundary and coboundary maps on a graph $G$. 

\begin{definition}[$\mathbb{Z}_2$-Boundary and coboundary maps]
    In this work we use binary vectors to indicate collections of vertices, edges, and cycles of the graph $G$. 
    The boundary map $\partial$ on an edge vector is a $\mathbb{Z}_2$-linear map that is defined by its action on a single edge $\partial e = v+v'$ where $v$ and $v'$ are the adjacent edges of $e$. 
    The coboundary map is given by the transpose of the boundary map $\delta = \partial^T$ and satisfies $\delta v = \sum_{e \ni v} e$. 
    Given a choice of a collection of cycles in $G$ we also define a second boundary map $\partial_2$ whose action on a cycle is $\partial_2 c = \sum_{e \in c} e$. 
    Similarly, we define a second coboundary map $\delta_2=\partial_2^T$ which acts on a single edge as $\delta_2 e = \sum_{c \ni e } c$.
\end{definition}

\begin{remark}
    \label{rem:Exactitude}
    The maps $\partial_2,$ and $\partial $ form an exact sequence if a generating set of cycles is chosen, and similarly for $\delta, \delta_2$. 
    These sequences are not short exact as $\delta$ has a nontrivial kernel given by all vertices. 
\end{remark}

Throughout this work we abuse notation by identifying the binary vector associated to a set of vertices, edges, or cycles with the set itself. 
Where this is done, the meaning should be clear from context.

\begin{proof}[Proof of Theorem~\ref{thm:GaugingMeasurement}]
Applying the gauging measurement to an initial state $\ket{\Psi}$ results in
\begin{align}
    \prod_v \frac{1}{2}(\mathds{1}+\varepsilon_v A_v) \ket{\Psi} \ket{0}_E
\end{align}
up to an overall normalization.
Here, $\varepsilon_v$ is the observed result of measuring the operator $A_v=X_v \prod_{e \ni v} X_e$ and $\ket{0}_E=\ket{0}^{\otimes E}$. 
Ungauging by measuring out the edge degrees of freedom in the $Z$ basis and discarding them results in the state 
\begin{align}
    \bra{z_e}_E\prod_v \frac{1}{2}(\mathds{1}+\varepsilon_v A_v) \ket{\Psi} \ket{0}_E ,
\end{align}
up to an overall normalization. 
Here, $\bra{z_e}_E=\otimes_e\bra{z_e}_e$ where $z_e$ is the observed result of measuring the operator $Z_e$. 
Next, we expand the product of projection operators into a sum
\begin{align}
    \bra{z_e}_E \frac{1}{2^V} \sum_{c \in C^0(G,\mathbb{Z}_2)} \varepsilon(c) X_V(c)X_E(\delta c) \ket{\Psi} \ket{0}_E ,
\end{align}
where the sum is over $\mathbb{Z}_2$-valued 0-cochains on $G$, 
\begin{align}
    \varepsilon(c)&=\prod_{v}\varepsilon_v^{c_v} ,
    \\
    X_V(c)&=\prod_v X_v^{c_v} ,
    \\
    X_E(\delta c)&=\prod_e X_e^{\{\delta c\}_e}.
\end{align}
Since ${\bra{z_e}_E X_E(\delta c) \ket{0}_E}$ is zero unless $z_e = \{\delta c\}_e$ on all edges, we can rewrite the ungauged state as 
\begin{align}
    \frac{1}{2^V} \sum_{c \text{ s.t. } \delta c = z} \varepsilon(c) X_V(c) \ket{\Psi}
    \nonumber \\
    =
    \frac{1}{2^V} X_V(c')\sum_{c \in Z^0(G,\mathbb{Z}_2)} \varepsilon(c) X_V(c) \ket{\Psi}
\end{align}
for a fixed $c'$ that satisfies $\delta c'=z$ and $Z^0(G,\mathbb{Z}_2)$ the group of 0-cocycles on $G$. 
For a connected graph $G$ there are only two elements $c\in Z^0(G,\mathbb{Z}_2)$, either $c_v=1$ for all vertices or $c_v=0$ for all vertices. 
Hence, the ungauged state is 
\begin{align}
    X_V(c') \frac{1}{2}  (\mathds{1}+\sigma L) \ket{\Psi}
\end{align}
where $L=\prod_v X_v$ is the logical operator being measured and $\sigma = \prod_v \epsilon_v$ is the observed outcome of the logical measurement.
Here, $X_V(c')$ is a Pauli byproduct operator determined by the observed measurement outcomes.

This byproduct operator is the same as the one defined in the final step of Algorithm~\ref{alg:GaugeLogical}
\end{proof}

\noindent\textbf{Graph desiderata and constructions} --- In this section we analyze the deformed code that is created after measuring the $A_v$ terms in the gauging measurement procedure. The auxiliary graph $G$ contains vertices $V$ and edges $E$.

\begin{lemma}[Deformed code]\label{lem:deformed}
    The following form a generating set of checks for the deformed code:
    \begin{itemize}
        \item $A_v = X_v \prod_{e \ni v} X_e$ for all $v\in V$ that are not dummy vertices.
        \item $A_v = \prod_{e \ni v} X_e$ for all dummy vertices $v\in V$. 
        \item $B_p = \prod_{e \in \gamma} Z_e$ for a generating set of cycles $\gamma\subseteq E$. 
        \item $\tilde{s}_i=s_i \prod_{e \in \mu_i } Z_e$ for all checks in the input code $s_i$ and appropriate perfect matchings $\mu_i\subseteq E$ for each of those checks.
    \end{itemize}
\end{lemma}

\begin{proof}
The $A_v$ operators become checks since they are measured during the gauging process. 
Both $B_p$ and $\tilde s_i$ operators can be thought of as arising from applying stabilizer update rules \cite{gottesman1997stabilizer} to the initial stabilizer group generated both by checks of the original code $\{s_i\}$ and the single-qubit $Z_e$ stabilizers on the edge qubits $e$.
Specifically, the $B_p$ operators originate from the $Z_e$ stabilizers because, for a product of these edge stabilizers to remain a check after measuring the $A_v=\prod_{e \ni v} X_e$ operators, it must commute with them all. 
This is equivalent to a product $\prod_{e\in\gamma}Z_e$ being over edges $\gamma$ that satisfies $\partial \gamma = 0$, i.e.~$\gamma$ represents a graph cycle. 
Similarly, the $\tilde{s}_i$ operators arise from products of $s_i$ and $Z_e$ stabilizers from the initialization step that commute with all $A_v$ operators. 
The specific $A_v$ operators that $s_i$ anticommutes with are those associated to vertices in the set $S_Z:=\text{supp}_Z(s_i)\cap\text{supp}(L)$, where $\text{supp}_Z(s_i)$ is the $Z$-type support of $s_i$.
Thus commuting with all $A_v$ requires a set of edges $\mu_i$ such that $\partial \mu_i = \text{supp}_Z(s_i)\cap\text{supp}(L)$, or in other words $\mu_i$ is a perfect matching of the vertices in $S_Z$.
\end{proof}

By counting the qubits and checks, one sees that the dimension of the codespace of the deformed code is only reduced by a single qubit compared with the original code, corresponding to the logical $L$ that is measured by the gauging deformation.
A total of $|E|$ new qubits are introduced along with $|V|$ new independent X-type checks and a number of new independent $Z$-type checks that is equal to $C$, the number of generating cycles of $G$. 
We then have the well-known relation $|E|-C-|V|=-1$ for a connected graph. 
This provides an alternative method to Theorem~\ref{thm:GaugingMeasurement} to see that no logical information is lost beyond the measurement of $L$. 

There is a large degree of freedom when specifying a generating set of checks in the deformed code. 
If we fix the choice of $A_v$ and $s_i$ checks, then the freedom can be associated to choosing cycles and matchings $\gamma$ and $\mu_i$. 
These choices do not affect the code space, as the $B_p$ operators for any generating set of cycles in $G$ generate the same algebra. 
Furthermore, the $\tilde{s}_j$ operator for two different choices of path $\gamma$ are related by multiplication with $B_p$ operators. 
In practice, we aim to choose a set of paths $\gamma$ and $\mu_i$ that result in a set of checks with small weight and that result in small qubit degree (the number of checks a qubit is involved in). 
Indeed, this is the point of the first three graph desiderata in Theorem~\ref{thm:GDesiderata}, which ensure the deformed code is LDPC.
The necessity and sufficiency of the desiderata for that purpose should now be clear from inspecting the full set of deformed code checks in Lemma~\ref{lem:deformed}.

To complete the proof of Theorem~\ref{thm:GDesiderata}, it remains to show that graph expansion implies the deformed code has code distance no smaller than the original code. We do that in the following lemma.

\begin{lemma}[Space fault-distance]
    \label{lem:spacedistance}
    The distance of the deformed code satisfies $d^* \geq \min(h(G),1) d$, where $h(G)$ is the Cheeger constant of $G$ and $d$ is the distance of the original code.
\end{lemma}

\begin{proof}
    A logical operator on the deformed code can be written as $L'=i^\sigma L_X^V  L_Z^V L_X^E L_Z^E \tilde{L}$. Here $\tilde{L}$ captures the support of $L'$ that does not intersect the gauged logical, $\mathcal{S}_X$ denotes the $X$ support of $L'$, ${L_X^V=\prod_{v\in \mathcal{S}_X \cap V_{G}} X_v}$ captures the intersection of the $X$ support of $L'$ with the gauged logical, ${L_X^E=\prod_{e\in \mathcal{S}_X\cap E_{{G}}} X_e}$ captures the $X$ support of $L'$ on the edges introduced by the gauging procedure, and similarly for $Z$. 

     First, we consider the $X$-type component of the logical operator $L'$.
    The logical operator must commute with the $B_p$ checks by definition, hence ${\mathcal{S}_X^E=\mathcal{S}_X\cap E_{{G}}}$ is a 1-cocycle on the graph $G$ as it satisfies ${\delta_2 \mathcal{S}_X^E = 0}$.
    Since the $B_p$ terms are defined on a generating set of cycles we have that the sequence formed by $\delta,\delta_2$ is exact, see Remark~\ref{rem:Exactitude}.
    Hence, $\mathcal{S}_X^E = \delta \tilde{\mathcal{S}}_X^V$ for some set of vertices $\tilde{\mathcal{S}}_X^V \subset V_{{G}}$. 
    From this we have ${L_X^E = \overline{L}_X^V \prod_{v \in \tilde{\mathcal{S}}_X^V} A_v}$, where ${\overline{L}_X^V = \prod_{v \in \tilde{\mathcal{S}}_X^V} X_v}$. 
    We now have another logical representative that is equivalent to $L'$, given by ${\overline{L}=i^{\sigma} L_X^V \overline{L}_X^V L_Z^V L_Z^E\tilde{L}=L'\prod_{v \in \tilde{\mathcal{S}}_X^V} }A_v$.

    Next, we consider the $Z$-type component of the logical operator $L'$.
    We first point out that the deformed checks $\tilde{s}_i$ are given by the original checks, potentially multiplied with some $Z_e$ operators. 
    Similarly, the equivalent logical operator $\overline{L}$ is some operator on the original qubits potentially multiplied by $Z_e$ operators.

    From this we can see that the equivalent logical operator restricted to the qubits of the original code ${\overline{L}|_V=i^\sigma  L_X^V \overline{L}_X^VL_Z^V \tilde{L}}$ must be a logical operator of the original code. 
    This is because it must commute with the deformed checks, and since the additional operators on the edge qubits in the deformed checks are all of the form $Z_e$ they play no role in the commutation relations. 
    From this we see that the full equivalent logical operator $\overline{L}$ is obtained from the restricted logical $\overline{L}|_V$ via the gauging code deformation. 

    Hence, any logical in the deformed code is equivalent to a logical on the original code $\overline{L}|_V$ potentially multiplied by some $Z_e$ operators. 
    The weight of any nontrivial logical on the original code, such as $\overline{L}|_V$, is lower bounded by the distance $d$. 
    Hence, the weight of the unrestricted logical $\overline{L}$ is also lower bounded by $d$. 
    Furthermore, we can construct $\overline{L}$ to have support on no more than half the vertex qubits $\leq \frac{|V|}{2}$ by optionally multiplying $L'$ with the stabilizer $\prod_{v\in V_{\bar{G}}} A_v$. 
    
    We now lower bound the relative change in operator weight, induced by the equivalence under multiplication with vertex stabilizers that convert the deformed logical $\overline{L}$ back to the logical $L'$, by the Cheeger constant $h(G)$ of a single layer of the graph $G$. 
    In the worst case, multiplication with $\prod_{v \in \tilde{\mathcal{S}}_X^V} A_v$ to convert $\overline{L}$ to $L'$ removes the support on the vertex set $\tilde{\mathcal{S}}_X^V$ and adds support on the edge coboundary set $\delta \tilde{\mathcal{S}}_X^V$. 
    Then we apply $|\delta \tilde{\mathcal{S}}_X^V|\geq h(G) | \tilde{\mathcal{S}}_X^V|$ to achieve the distance bound on the logical $L'$. 
\end{proof}

In the rest of this subsection, we discuss the cycle sparsification step in more detail and ensure it preserves the first three desiderata and the deformed code distance as claimed in Theorem~\ref{thm:WorstCaseG}.

\begin{definition}[Cycle-sparsified graph $\bar{G}$]
    \label{def:cycle-sparsification}
    Given a graph $G$ and a constant $c$ we call the cycle-degree, a cycle-sparsification of $G$ is a new graph that is built by adding edges to copies of $G$ numbered $1,2\dots,r$, and alternatively referred to as \textit{layers} (as in Fig.~\ref{fig:thick_Tanner}, see also Fig.~\ref{fig:cycle_sparse}).
    One type of additional edges connect each vertex to its layer in the subsequent layer and serve to implement the cartesian product operation of graph $G$ with a path graph on $r$ vertices.  
    The other type of additional edges cellulate a cycle into triangles by connecting vertices as follows $\{(1,N-1), (N-1,2), (2,N-2), (N-2,3), \dots \}$ following an ordering of the vertices as they are visited when the cycle is traversed, see Figure~\ref{fig:CycleWeightReduction}. 
    If the original graph had cycle basis $\{\gamma_i\}$, the cycle-sparsified graph cellulates exactly one copy of each $\gamma_i$, and, given a sufficient number of layers $r$, the number of these cellulated cycles containing any particular edge can be chosen to be at most $c$ for any choice of $c\ge1$.
    We use $R_G^c$ to denote the minimal number of layers to achieve a cycle-sparsification of $G$ with constant $c$. Besides $|L|$ number of specified vertices in the 1st layer, all vertices in the cycle sparsified graph are dummy vertices.
\end{definition}

\begin{remark}
    The cycle-sparsified graph $\bar{G}$ may fail to satisfy desideratum 4. 
    This is because the cycle sparsification step does not preserve $h(G)\ge1$. In particular, once there are more than four copies of $G$ in the cycle-sparsified graph $\bar{G}$ the set $V_{0,1}=V_{G_0}\cup V_{G_1}$ of \textit{all} vertices on the first two copies of $G$ satisfies $|\delta V_{0,1}| = | V_{G_1}| = \frac{1}{2}|V_{0,1}|$ and hence the Cheeger constant is upper bounded by a half. 
    Below, we prove that $\bar{G}$ nonetheless results in a deformed code with the same distance bound as for $G$. 
    The reason behind this is that sufficient expansion in the first layer of $G$ that is attached directly to the logical operator suffices to guarantee the distance is preserved during the code deformation.  
\end{remark}

\begin{corollary}[Cycle-sparsified space fault-distance]
\label{cor:SparsifiedDistance}
    The distance of the deformed code based on the cycle-sparsified graph $\bar{G}$, see Definition~\ref{def:cycle-sparsification}, satisfies $d^*\geq (h(G),1) d$. 
\end{corollary}

\begin{proof}
Following the proof of Lemma~\ref{lem:spacedistance}, we have ${L_X^E = \overline{L}_X^V \prod_{v \in \tilde{\mathcal{S}}_X^V} A_v}$, where ${\overline{L}_X^V = \prod_{v \in \tilde{\mathcal{S}}_X^V \cap V_{G_0}} X_v}$. 
In the definition of $\overline{L}_X^V$, the product of operators is only over vertices in $G_0$ since the vertices in other layers are dummy vertices which do not support qubits. 
We now have that ${\overline{L}=i^{\sigma} L_X^V \overline{L}_X^V L_Z^V L_Z^E\tilde{L}=L'\prod_{v \in \tilde{\mathcal{S}}_X^V} }A_v$ is a logical representative equivalent to $L'$. 
The restriction of the logical $\overline{L}$ to the original qubits produces a logical in the original code, hence the weight of $\overline{L}$ is at least the distance of the original code.

To bound the distance of the general logical $L'$ in the deformed code, we now focus on the change in weight caused by the term $\prod_{v \in \tilde{\mathcal{S}}_X^V} A_v$. 
The size of the vertex set $\tilde{\mathcal{S}}_X^V\cap V_{G_0}$ can be taken $\leq \frac{|V_{G_0}|}{2}$, up to multiplying $L'$ with the measured logical which is now the stabilizer $\prod_{v\in\bar{G}} A_v$. 
The operator $\overline{L} \prod_{v \in \tilde{\mathcal{S}}_X^V} A_v$ has $X$ support on at least all edges in $\big(\delta (\tilde{\mathcal{S}}_X^V\cap V_{G_0}) \big) \cap E_{G_0}$. 
This is the set of edges inside $G_0$ that are in the coboundary of the relevant vertices in $G_0$. 
The size of this set satisfies ${|\big(\delta (\tilde{\mathcal{S}}_X^V\cap V_{G_0}) \big) \cap E_{G_0}|\geq h(G) |\tilde{\mathcal{S}}_X^V\cap V_{G_0} |}$. 
Hence, the relative weight of $L'$ and $\overline{L}$ is lower bounded by $h(G)$. 
\end{proof}

\begin{figure}[t]

\subfloat[Cycle sparsification]{\raisebox{3pt}%
{\includegraphics[width=0.48\columnwidth]{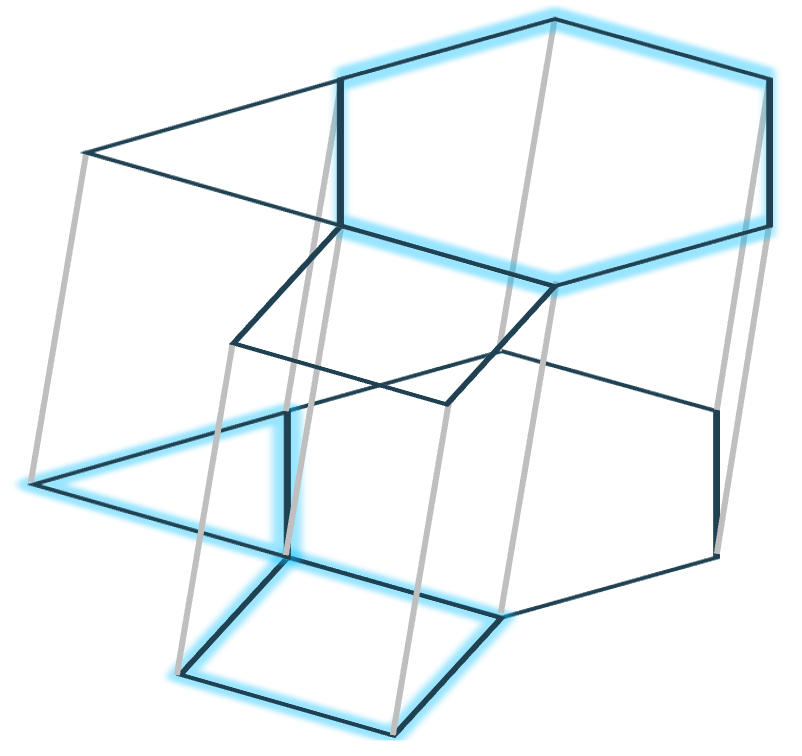}}
\label{fig:cycle_sparse}
}\hspace*{\fill}
\subfloat[Cellulation]{\raisebox{3pt}%
{\includegraphics[width=0.48\columnwidth]{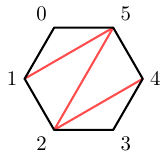}}
\label{fig:CycleWeightReduction}
}

\caption{\textbf{Cycle sparsification and cellulation:} \textbf{(a)} Cycle sparsification with two layers. Black edges are copies of the original graph and gray edges connect them. A complete cycle basis consists of the highlighted cycles and all length-4 cycles between layers that traverse exactly two black and two gray edges. The constant $c$ in Definition~\ref{def:cycle-sparsification} is 1 after this sparsification, whereas it was 2 for the original graph, and therefore $R^1_G=2$. \textbf{(b)} Cellulating a weight-six cycle (black) into triangles by adding additional edges (red). }
\end{figure}

\begin{lemma}
    \label{lem:SparsifiedDesiderata}
    Given a graph $G$ that satisfies the first two desiderata in Theorem~\ref{thm:GDesiderata}, its cycle-sparsification $\bar{G}$ satisfies the first three desiderata in Theorem~\ref{thm:GDesiderata}
\end{lemma}

\begin{proof}
    The degree of the cycle-sparsified graph $\bar{G}$ is upper bounded by the degree of $G$ plus the constant ${(\text{deg}(G)c+2)}$, where $c$ is the congestion number of cycles in the chosen basis that are assigned to a single layer of $\bar{G}$. 
    Hence, if the graph $G$ satisfies desideratum 1, so too does $\bar{G}$. 
    The cycle-sparsified graph $\bar{G}$ supports matchings $\mu_i$ that can be routed entirely through the first layer since that layer is already a connected graph containing all the non-dummy vertices. 
    Hence, assuming the graph $G$ satisfies desideratum 2, so too does $\bar{G}$. 
    The cycle-sparsified graph $\bar{G}$ has a cycle basis consisting of length-3 and 4 cycles by construction. 
    Hence, $\bar{G}$ satisfies desideratum 3. 
\end{proof}

We use the freedom in choosing checks of the deformed code per Lemma~\ref{lem:deformed} for the cycle-sparsified graph to take advantage of the layered structure. 
This leads to a very natural cycle basis constructed from cycles with length $\le4$ coming in two types. 
First, for each layer $i=1,2,\dots,r-1$ and each edge $e\in E$ of the original graph, there is a length-4 cycle $\gamma_{i,e}$, which consists of the copy of edge $e$ in layer $i$, the copy of the edge $e$ in layer $i+1$, and the two edges connecting the layers together and adjacent to those copies of $e$. 
Second, there are length-3 cycles resulting from the triangular cellulations.

In Definition~\ref{def:cycle-sparsification} we have chosen to cellulate the cycles into triangles as they have minimal weight. 
A similar procedure applies for squares, or even arbitrary polygons, which need not have a uniform number of edges. 
We note that using squares is also natural in the sense that square cycles already appear between layers in the cycle-sparsified graph.

For a constant degree graph there are $\Theta( |V|)$ cycles in a minimal generating set. 
For a random expander graph, making an appropriate choice, almost all generating cycles are expected to be of length $O(\log|V|)$. 
In this case we expect a cycle-degree of $O(\log|V|)$ and a number of layers $R_G^c=O(\log|V|)$ required for cycle-sparsification.
We remark that the decongestion lemma~\cite{freedman2021building} establishes a worst-case bound $R_G^c=O(\log^3|V|)$ for cycle-sparsification of a constant degree graph.
This leads to the overhead scaling bound in Theorem~\ref{thm:WorstCaseG}, see Ref.~\cite{he2025extractors} for a detailed review of this point. 
On the other hand, for some cases of the gauging measurement no sparsification is required, such as the BB code example presented later. 
It would be interesting to understand conditions under which this occurs generally.

\noindent\textbf{Fault-tolerant implementation} --- 
In this section we provide an overview of the fault tolerant operation of the gauging measurement procedure. 
Detailed definitions and proofs of fault tolerance are provided in the spacetime code and fault distance section of the Supplementary information~\cite{supplement}.

To prove Theorem~\ref{thm:fault-tolerance}, we analyze gauging measurement in a standard phenomenological noise model \cite{dennis2002topological} in which qubits and measurements of Pauli operators can suffer errors, but we ignore circuit-level details that may spread errors or introduce other sources of correlated errors. Our main strategy to ensure fault tolerance is to repeat the measurement of the checks of the deformed code, i.e.~$A_v$, $B_p$, $\tilde s_i$, at least $d$ times, where $d$ is the code distance of the original code. Intuitively, this provides redundant, and hopefully independent, measurements of logical $L$ since the set of $A_v$ measurements at any time $t=1,2,\dots,d$, denoted $\{A^{(t)}_v\}$, would ideally satisfy $L=L^{(t)}:=\prod_v A^{(t)}_v$. If the $L^{(t)}$ values are not constant in time, however, we know an error occurred and can, in some cases, correct it.

The complete proof of Theorem~\ref{thm:fault-tolerance} is involved, so we leave the details to the spacetime code and fault distance section of the supplementary information~\cite{supplement} and only summarize the main ideas here. This starts with some basic definitions (similar to Ref.~\cite{McEwen2023}).

\begin{definition}[Space and time faults]
    A space-fault is a Pauli error operator that occurs on some qubit during the implementation of the gauging measurement.
    A time-fault is a measurement error where the result of a measurement is reported incorrectly during the implementation of the gauging measurement.
    A general spacetime fault is a collection of space and time faults. 
\end{definition}

For the purposes of our discussion, state initializations can be implemented via single qubit $Z$ measurement followed by a conditional application of $X$ to prepare the $\ket{0}$ state. 
State initialization faults (occurring for example when and edge qubit is initialized in $\ket{1}$ instead of $\ket{0}$ due to an error) can be considered space-faults by decomposing them into a perfect initialization followed by a Pauli error.

\begin{definition}[Detectors]
    A detector is a collection of state initializations and measurements that ideally yield a deterministic result, in the sense that the product of the observed measurement results must be $+1$ independent of the individual measurement outcomes, if there are no faults in the procedure.
\end{definition}

The set of detectors defines a \emph{spacetime code}. We enumerate the detectors in gauging measurement in supplementary Lemma~1. Generally, detectors are formed from subsequent measurements of the same check operator. This changes slightly at the first round of measurement in the deformed code, where instead the measurements of $s_i$ and $\tilde s_i$ together form a detector, and upon the return to the original code when the edge qubits are measured out, where instead measurements of $s_i$, $\tilde s_i$, and of the edge qubits in $\mu_i$ altogether form a detector.

\begin{definition}[Syndrome]
    The syndrome caused by a spacetime fault is defined to be the set of detectors that are violated in the presence of the fault. 
    That is, the set of detectors that do not satisfy the constraint that the observed measurement results multiply to $+1$ in the presence of the fault. 
\end{definition}

With these definitions in hand, it is natural to define a \emph{spacetime logical fault} to be a collection of space and time faults that violates no detectors, or in other words, has a trivial syndrome. Similarly, by propagating the errors through the procedure, we can determine if a spacetime logical fault affects the result of the gauging measurement procedure, either by flipping the logical measurement result, or by causing a logical operator other than $L$ to be applied. If the spacetime logical fault does not affect the logical outcome of the procedure, we say it is a \emph{spacetime stabilizer}, and otherwise that it is a nontrivial spacetime logical fault. The minimum number of faults in a nontrivial spacetime logical fault is called the \emph{spacetime fault-distance}.

In the spacetime code and fault distance section of the supplementary information~\cite{supplement}, supplementary Theorem~1 shows that the spacetime fault-distance of gauging measurement is $d$ provided that a graph $G$ satisfying desideratum 3 in Theorem~\ref{thm:GDesiderata} is used, and that at least $d$ rounds of syndrome measurement in the deformed code are performed. Interestingly, we can show this by separately showing (1) that it takes at least $d$ measurement faults to cause a nontrivial spacetime logical fault (i.e.~the \emph{time fault-distance} is at least $d$, see supplementary Lemma~3) and (2) that the distance of the deformed code, or what we called the \emph{space fault-distance} in Lemma~\ref{lem:spacedistance}, is also at least $d$. Effectively, these two failure mechanisms, time-like or space-like faults, can be cleanly separated (see supplementary Lemma~4) by multiplying an arbitrary spacetime logical fault by spacetime stabilizers, the structure of which we also explicitly describe in supplementary Lemma~2.

\section*{Data Availability}
No data was generated during this work. 

\clearpage 
\onecolumngrid

\begin{center}
    \noindent\textbf{\Large{Supplementary information for ``Low-overhead fault-tolerant quantum computation by gauging logical operators''}}
\end{center}

\section*{BB code examples} 
To demonstrate the utility of gauging measurement we provide efficient gauging measurements for a pair of bivariate bicycle (BB) codes introduced in Ref.~\cite{bravyi2024high}. Let $I_r$ be the $r\times r$ identity matrix and $C_r$ be the cyclic permutation matrix of $r$ items, i.e.~$\langle i|C_r=\langle i+1\mod r|$ for all $i$. Define $x=C_{\ell}\otimes I_m$ and $y=I_{\ell}\otimes C_m$. Note $x^\ell=y^m=I_{\ell m}$ and $x^\top x=y^\top y=I_{\ell m}$.

BB codes are CSS codes using $n=2\ell m$ physical qubits divided into two sets, $\ell m$ left ($L$) qubits and $\ell m$ right ($R$) qubits. The parity check matrices of a BB code are
\begin{equation}
H_X=[A|B],\quad H_Z=[B^\top|A^\top],
\end{equation}
where $A,B\in\mathbb{F}_2[x,y]$ are polynomials in the matrices $x$ and $y$. Also, $A^\top=A(x,y)^\top=A(x^\top,y^\top)=A(x^{-1},y^{-1})$ and likewise for $B^\top$. There are $\ell m$ $X$ checks and $\ell m$ $Z$ checks in this generating set, though only $(n-k)/2$ checks of either type are independent for a code encoding $k$ qubits.

The $X$ checks, $Z$ checks, $L$ qubits, and $R$ qubits are each in 1-1 correspondence with the $\ell m$ elements of
\begin{equation}
\mathcal{M}=\{x^ay^b:a,b\in\mathbb{Z}\}.
\end{equation}
To label individual checks or qubits, we write $(\alpha,T)$ for $\alpha\in\mathcal{M}$ and $T\in\{X,Z,L,R\}$. A set of checks or qubits is more generally indicated by $(p,T)$ for a polynomial ${p\in\mathbb{F}_2[x,y]}$. For instance, by definition of the code, the $X$ check labeled $(\alpha,X)$ acts on qubits $(\alpha A,L)$ and $(\alpha B,R)$ while $Z$ check labeled $(\beta,Z)$ acts on qubits $(\beta B^\top,L)$ and $(\beta A^\top,R)$. 

An $X$-type (resp.~$Z$-type) Pauli acting on left qubits $(p,L)$ and right qubits $(q,R)$ for polynomials $p,q$ is more succinctly written $X(p,q)$ (resp.~$Z(p,q)$). For example, the checks described previously are written $X(\alpha A,\alpha B)$ and $Z(\beta B^\top,\beta A^\top)$.

\subsection*{Gross code} 
The gross code is a $\llbracket144,12,12\rrbracket$ BB code, and is so-called because a ``gross" is a dozen dozens or 144. It is obtained from the BB construction by choosing $\ell=12$, $m=6$, and
\begin{align}
A=x^3+y^2+y,\quad B=y^3+x^2+x.
\end{align}
We use $A_i$ and $B_i$ for $i=1,2,3$ to denote the individual monomial terms in these polynomials.

A convenient basis of logical operators for the gross code was provided in Ref.~\cite{bravyi2024high}. This basis is described using polynomials
\begin{align}
f=&1+x+x^2+x^3+x^6+x^7+x^8+x^9\\\nonumber
&+(x+x^5+x^7+x^{11})y^3,\\\nonumber
g=&x+x^2y+(1+x)y^2+x^2y^3+y^4,\\\nonumber
h=&1+(1+x)y+y^2+(1+x)y^3.
\end{align}
Then $\overline{X}_{\alpha}=X(\alpha f,0)$ and $\overline{X}'_{\beta}=X(\beta g,\beta h)$ are logical operators for all choices of monomials $\alpha,\beta\in\mathcal{M}$. The symmetry in the BB code construction implies ${\overline{Z}_{\beta}=Z(\beta h^\top,\beta g^\top)}$ and $\overline{Z}'_{\alpha}=Z(0,\alpha f^\top)$ are also logical operators. This symmetry also means that if we construct gauging measurements for $\overline{X}_\alpha$ and $\overline{X}'_\beta$, the same Tanner graph connectivity works for $\overline{Z}'_\alpha$ and $\overline{Z}_\beta$, a fact that was used previously to reduce the number of ancilla systems required to measure a complete logical basis \cite{cross2024linear}.

It was noted before \cite{cross2024linear,Cowtan2024} that, while $\overline{X}'_\beta$ can be measured by mono-layer versions of the CKBB method~\cite{cohen2022low}, $\overline{X}_\alpha$ requires several layers. This can be understood as the Tanner subgraph supported on $\overline{X}_\alpha$ lacking sufficient expansion \cite{cross2024linear}. 

The flexibility of gauging measurement allows us to measure $\overline{X}_\alpha$ with fewer additional qubits and checks than existing constructions. Our goal is to measure $\overline{X}_\alpha$ while not introducing any new checks or qubits with Tanner graph degree more than $6$ (note the 12 qubits in $\overline{X}_\alpha$ and the 18 adjacent $Z$ checks will necessarily become degree $7$). It happens that we can achieve this goal without decongestion or cellulation, i.e.~Fig.~1.~(a) in the main text is sufficient. To specify the construction, we describe $G,M,N$.

The vertices of graph $G$ are in 1-1 correspondence with the qubits $\{(\gamma,L):\gamma\in f\}$, where $\gamma\in f$ means monomial $\gamma$ is a term in polynomial $f$. To ensure that each row of the matrix $M$ (see Fig.~1.~(a) in the main text) has weight $1$, we connect two vertices $\gamma,\delta\in f$ of $G$ if qubits $(\gamma,L)$ and $(\delta,L)$ participate in the same $Z$ check. This is the case if and only if $\gamma=B_i^\top B_j\delta$ for some $i,j\in\{1,2,3\}$. From this step, $G$ acquires $18$ edges and is the same as it would be in the CKBB method \cite{cohen2022low}.

Additional edges can now be added to $G$ to increase its expansion. This is done to ensure the deformed code has code distance equal to the original code, here distance 12. Constructing a graph with large Cheeger constant is sufficient but not necessary for this objective, so we instead randomly add edges to $G$. It is typically fast to eliminate random trials with low code distances using upper bounds provided by the BP+OSD decoder (as described in Ref.~\cite{bravyi2024high}). If a trial passes the BP+OSD test, then the deformed code distance can be proven to be 12 exactly with integer programming \cite{cplex2022v22}.

We find that four additional edges are sufficient to make a deformed code with distance 12. With vertices labeled by monomials from $f$, one choice of additional edges is
\begin{equation}
\begin{array}{cccc}
(x^2,x^5y^3),& (x^2,x^6),& (x^5y^3,x^{11}y^3),& (x^7y^3,x^{11}y^3).
\end{array}
\end{equation}

As a connected graph with $12$ vertices and $22$ edges, $G$ has a minimal cycle basis consisting of $22-12+1=11$ cycles. Gaussian elimination can be used to find such a basis as the row nullspace of $G^\top$, i.e.~find full-rank $N$ with $11$ rows such that $NG^\top=0$. The rows of $N$ represent $Z$ checks, the flux operators $\mathcal{B}$. However, because the $Z$ checks of a BB code are not all independent, these $11$ flux operators also may not all be independent. Suppose $S$ is the sub-matrix of $H_Z$ containing just the rows describing $Z$ checks in $\mathcal{S}$ (i.e.~checks overlapping $\overline{X}_\alpha$) and $C$ is the matrix describing the rest of the $Z$ checks. If $uS+vC=0$ for some vectors $u,v$, then $u,v$ identify a product of checks from sets $\mathcal{S}$ and $\mathcal{C}$ in the deformed code that is not supported on the original code and so must represent a cycle $uM$ supported on qubits $\mathcal{E}$ (possibly an empty cycle if $u=0$). Let $U=\{u:\exists v,uS+vC=0\}$. The number of redundant cycles in any cycle basis is then
\begin{equation}
\dim U=\mathrm{row\_nullity}(H_Z)-\mathrm{row\_nullity}(C),
\end{equation}
which evaluates to $4$ for the case at hand. Note this calculation would change if multiple logical operators were undergoing measurement simultaneously.

We therefore just need $11-4=7$ $Z$ checks in $\mathcal{B}$. We describe these checks as cycles on the vertices of $G$.
\begin{align}\nonumber
&x^9\rightarrow x^7y^3\rightarrow x^8\rightarrow x^9,\\\nonumber
&x^9\rightarrow x^{11}y^3\rightarrow x^7y^3\rightarrow x^9,\\\nonumber
&x^6\rightarrow x^7\rightarrow x^5y^3\rightarrow x^6,\\\nonumber
&x^6\rightarrow x^2\rightarrow x^5y^3\rightarrow x^6,\\\nonumber
&x^3\rightarrow x^2\rightarrow x^5y^3\rightarrow x^3,\\\nonumber
&x^6\rightarrow x^7\rightarrow x^8\rightarrow x^7y^3\rightarrow x^6,\\
&x\rightarrow x^2\rightarrow x^5y^3\rightarrow x^{11}y^3\rightarrow x.
\end{align}
This completes the description of the gauging measurement of $\overline{X}_\alpha$ in the gross code.

To summarize, Table~\ref{tab:gross_check_qubit_degrees} lists check weights and qubit degrees in the deformed code. Beyond the checks and qubits of the original gross code, the deformation adds $12$ $X$ checks in $\mathcal{A}$, $7$ $Z$ checks in $\mathcal{B}$, and $22$ qubits in $\mathcal{E}$. The additional checks and qubits total $41$.

\begin{table}[t]
    \centering
    \begin{tabular}{|c|c|}
        \hline
        \multicolumn{2}{|c|}{$X$ checks} \\\hline
        weight & count \\\hline\hline
        4 & 7 \\
        5 & 2 \\
        6 & 75 \\\hline
        \multicolumn{2}{c}{}\\
        \multicolumn{2}{c}{}
    \end{tabular}
    \begin{tabular}{|c|c|}
        \hline
        \multicolumn{2}{|c|}{$Z$ checks}\\\hline
        weight & count \\\hline\hline
        3 & 5 \\
        4 & 2 \\
        6 & 54 \\
        7 & 18 \\\hline
        \multicolumn{2}{c}{}
    \end{tabular}
    \begin{tabular}{|c|c|}
        \hline
        \multicolumn{2}{|c|}{qubits} \\\hline
        degree & count \\\hline\hline
        3 & 8 \\
        4 & 9 \\
        5 & 5 \\
        6 & 132 \\
        7 & 12 \\\hline
    \end{tabular}
    \caption{Check weights and qubit degrees of the deformed code for measuring $\overline{X}_{\alpha}$ in the gross code.}
    \label{tab:gross_check_qubit_degrees}
\end{table}

To demonstrate the versatility of gauging measurement, we provide another BB code example below. This is the $\llbracket288,12,18\rrbracket$ code from Ref.~\cite{bravyi2024high} that we refer to as the two-gross code.

\subsection*{Two-gross code}

\textit{Two-gross code.} Twice as large as the gross code but with larger code distance, the $\llbracket 288,12,18\rrbracket$ two-gross code is obtained from the BB code construction by taking $\ell,m=12,12$, and
\begin{equation}
A=x^3+y^7+y^2,\quad B=y^3+x^2+x.
\end{equation}
One set of logical operators for the two-gross code are the weight-18 operators $\overline{X}_\alpha=X(\alpha f,0)$ for
\begin{align}\nonumber
f=&1+x+x^2+x^7+x^8+x^9+x^{10}+x^{11}\\\nonumber
&+(1+x^6+x^8+x^{10})y^3\\
&+(x^5+x^6+x^9+x^{10})y^6+(x^4+x^8)y^9
\end{align}
and all $\alpha\in\mathcal{M}$. 

In order to measure $\overline{X}_\alpha=X(\alpha f,0)$ with gauging measurement, we use the construction of Fig.~1.~(a) in the main text and construct a graph $G$ with suitable properties. Just as for the gross code, each vertex $\gamma\in\mathcal{M}$ of the graph corresponds to a monomial term in $f$. An edge $(\gamma,\delta)$ is added if $\gamma=B_i^\top B_j\delta$ for some monomials $B_i,B_j\in B$ to ensure a sparse matching matrix $M$. A total of $27$ edges are added this way.

We find we can also add an additional 7 edges to ensure the deformed code has distance $18$. Two of these edges connect the same two vertices, so the resulting graph is a multi-graph.
\begin{align}
\begin{array}{lll}
(x^{4}y^{9},x^{9}y^{6}),&(y^{3},x^{11}),&(x^{7},x^{10}y^{6}),\\
(x^{8}y^{3},x^{10}y^{6}),&(1,x^{8}),&(x^{2},x^{6}y^{3})\text{\space twice}.
\end{array}
\end{align}
We can then find a basis of $13$ cycles to complete the construction.

\begin{align}\nonumber
&x^{2}\rightarrow x^{4}y^{9}\rightarrow x^{9}y^{6}\rightarrow x^{10}y^{6}\rightarrow x^{7}\rightarrow x^{6}y^{3}\rightarrow x^{2}\\\nonumber 
&x^{4}y^{9}\rightarrow x^{6}y^{6}\rightarrow x^{8}y^{3}\rightarrow x^{10}y^{6}\rightarrow x^{8}y^{9}\rightarrow x^{9}y^{6}\rightarrow x^{4}y^{9}\\\nonumber 
&1\rightarrow x\rightarrow y^{3}\rightarrow x^{10}y^{6}\rightarrow x^{7}\rightarrow x^{8}\rightarrow 1\\\nonumber 
&1\rightarrow x\rightarrow x^{2}\rightarrow x^{6}y^{3}\rightarrow x^{8}\rightarrow 1\\\nonumber 
&1\rightarrow x\rightarrow x^{2}\rightarrow y^{3}\rightarrow x^{11}\rightarrow 1\\\nonumber 
&1\rightarrow x^{8}\rightarrow x^{9}\rightarrow x^{10}\rightarrow x^{11}\rightarrow 1\\\nonumber 
&x^{2}\rightarrow x^{6}y^{3}\rightarrow x^{5}y^{6}\rightarrow x^{4}y^{9}\rightarrow x^{2}\\\nonumber 
&x^{4}y^{9}\rightarrow x^{5}y^{6}\rightarrow x^{6}y^{6}\rightarrow x^{4}y^{9}\\\nonumber 
&x^{7}\rightarrow x^{8}\rightarrow x^{6}y^{3}\rightarrow x^{7}\\\nonumber 
&x^{9}\rightarrow x^{10}\rightarrow x^{8}y^{3}\rightarrow x^{9}\\\nonumber 
&1\rightarrow x^{11}\rightarrow x^{10}y^{3}\rightarrow 1\\\nonumber 
&x^{7}\rightarrow x^{8}y^{9}\rightarrow x^{10}y^{6}\rightarrow x^{7}\\ 
&x^{2}\rightarrow x^{6}y^{3}\rightarrow x^{2} 
\end{align}

Check weights and qubit degrees for this two-gross code measurement are summarized in Table~\ref{tab:doublEross_check_qubit_degrees}. As for the gross code, we find a deformed code with the minimal number of weight 7 checks and degree 7 qubits while requiring no higher connectivity. The deformation adds $18$ $X$ checks, $13$ $Z$ checks, and $34$ qubits, totaling $65$.

\begin{table}[t]
    \centering
    \begin{tabular}{|c|c|}
        \hline
        \multicolumn{2}{|c|}{$X$ checks} \\\hline
        weight & count \\\hline\hline
        4 & 7 \\
        5 & 8 \\
        6 & 147 \\\hline
        \multicolumn{2}{c}{}\\
        \multicolumn{2}{c}{}\\
        \multicolumn{2}{c}{}
    \end{tabular}
    \begin{tabular}{|c|c|}
        \hline
        \multicolumn{2}{|c|}{$Z$ checks}\\\hline
        weight & count \\\hline\hline
        2 & 1 \\
        3 & 5 \\
        4 & 1 \\
        5 & 3 \\
        6 & 120\\
        7 & 27 \\\hline
    \end{tabular}
    \begin{tabular}{|c|c|}
        \hline
        \multicolumn{2}{|c|}{qubits} \\\hline
        degree & count \\\hline\hline
        3 & 3 \\
        4 & 17 \\
        5 & 12 \\
        6 & 272 \\
        7 & 18 \\\hline
        \multicolumn{2}{c}{}
    \end{tabular}
    \caption{Check weights and qubit degrees of the deformed code for measuring $\overline{X}_{\alpha}$ in the two-gross code.}
    \label{tab:doublEross_check_qubit_degrees}
\end{table}

\section*{Recovering existing protocols from the gauging measurement framework}
In this section we discuss how several existing schemes for logical measurement are related to gauging measurements. 

Lattice surgery is a widely used scheme for logical measurements on surface codes~\cite{horsman2012surface}. 
The gauging measurement can be interpreted as a direct generalization of lattice surgery as it recovers conventional lattice surgery when applied to copies of the surface code with an appropriate choice of $G$. 
For example, consider the logical operator $\overline{X}_1\otimes \overline{X}_2$ supported on the right and left edge of a pair of equally sized surface code blocks, respectively. 
Applying the gauging measurement with a choice of graph that is a ladder joining the edge qubits of the surface codes as shown in Fig.~1.~(b) in the main text results in a deformed code that is again the surface code on the union of the two patches. 
The final step of measuring out individual edges is the same as conventional lattice surgery.
To implement a lattice surgery between surface codes that are not directly adjacent to one another, we can apply the gauging procedure with a graph that includes a grid of dummy vertices between the two edges.

Interestingly, this procedure can be extended directly to measure any pair of matching logical $X$ operators on a pair of code blocks. 
The logical $X$ operators are further required to each have the same choice of graph $G$ which is allowed to have low expansion $h(G)<1$ but is required to satisfy the remaining two desiderata in Theorem~2 from the main text for each code block. 
Additional \textit{bridge} edges are then added between vertices in the two copies of $G$, generalizing the choice shown in Fig.~1.~(b) in the main text. 
Similar to lattice surgery, the gauging measurement defined by such a choice of graph preserves the code distance when the individual logical $X$ operators have minimal weight and contain no sublogical operators. 
We remark that this procedure is similar to the joint logical measurement constructions in Refs.~\cite{Cowtan2023,Cowtan2024,cross2024linear,Swaroop2024}. 
This discussion demonstrates that the expansion condition we have used in this work is overkill in some settings. 
It appears that expansion is only required for subsets of qubits that are relevant to the logical operators of the codes being measured. 
Developing a better understanding of the necessary expansion conditions is an interesting direction that might lead to more efficient logical measurement protocols.

Shor-style logical measurement~\cite{shor1996fault} involves entangling an auxiliary GHZ state to a code block via $CX$ gates that are applied transversally between the auxiliary qubits and the support of an $X$ logical. 
The logical is then measured by measuring $X$ on each of the auxiliary qubits and discarding them. 
We can perform a similar gauging measurement using a graph that has a separate dummy vertex connected by an edge to each qubit in the support of $L$, and then a connected graph on the dummy vertices. 
If we consider performing the gauging measurement where the edges of the connected graph on the dummy qubits are measured first, we are left with a state that corresponds to a GHZ state entangled with the support of $L$. This is similar to the Shor-style measurement where the final $X$ measurements have been commuted backwards through the $CX$ gates. 

A generalized version of the gauging measurement allows one to replace the auxiliary graph with a hypergraph whose adjacency matrix has the same kernel, see Remark~\ref{rem:Hypergraph}. 
This is sufficient to capture the Cohen et al.~scheme for logical measurement~\cite{cohen2022low}. 
Consider the restriction of the $Z$-type checks to the support of an irreducible $X$ logical as was done in Ref.~\cite{cohen2022low}. 
This defines a hypergraph of $Z$ constraints with the only nontrivial element in the kernel being the logical operator that is to be measured. 
Next, we add $d$ layers of dummy vertices for each qubit in the support of $L$, connect the $d$ copies of each vertex via a line graph, and join the vertices in each layer via a copy of the same underlying hypergraph. 
Applying the generalized gauging procedure to this hypergraph exactly reproduces the Cohen et al.~measurement scheme. 
The logical measurement scheme from Ref.~\cite{cross2024linear} can similarly be recovered by using fewer than $d$ layers of dummy vertices above.  
The procedures in Refs.~\cite{cohen2022low,cross2024linear} for joining ancilla systems designed for irreducible logicals to measure their products can also be captured as a gauging measurement by adding edges between the graphs corresponding to the individual ancilla systems. 
 
Generalizing the gauging measurement to a hypergraph allows one to measure many logical operators simultaneously, see Remark~\ref{rem:Hypergraph}. 
This approach can be used to implement the standard initialization of a CSS code by preparing $\ket{0}^{\otimes n}$ and measuring the $X$-type checks. 
This is achieved by starting with a trivial code with a dummy vertex for each $X$-type check of the CSS code and then performing the generalized gauging measurement using the hypergraph corresponding to the $Z$-type checks of the CSS code. 
Previous work on gauging quantum codes focused on this type of initialization measurement~\cite{Williamson2020a,Tantivasadakarn2021,Tantivasadakarn2022,Williamson2016,Vijay2016,kubica2018ungauging,Dolev2021,Rakovszky2023}. 
In this case the ungauging step simply performs a read-out measurement of $Z$ on all qubits. 

This state preparation and read-out gauging measurement procedure can be combined with another gauging measurement to implement a Steane-style measurement of a stabilizer group~\cite{Steane1996Active}.
This can be achieved by first performing state preparation of an ancilla code block via gauging as described above. 
This is followed by a gauging measurement of $XX$ on pairs of matching qubits between the data code block and the ancilla code block. 
Finally, the ungauging step is performed to read-out $Z$ on all ancilla qubits.

\section*{Generalizations of the gauging measurement procedure}

The gauging measurement procedure is vastly generalizable.
It can be applied to any representation of a finite group by operators that have a tensor product factorization. 
This allows vertex checks to connect to several qubits of the original code at once. 
With this approach, one can measure a weight $d$ logical operator using an auxiliary graph with just $d/c$ vertices instead of $d$ by connecting each vertex to a different set of $c$ qubits in the logical support. For instance, the surface code lattice surgery in Fig.~1.~(b) in the main text could then be done in the more standard way with one new column and $c=2$. Generally, the graph desiderata need to also be modified to accommodate this change. In particular, one may need more expansion, $h(G)\ge c$.  

The above generalization accommodates non-Pauli operators, whose measurement can produce magic states. 
An example of this is the measurement of Clifford operators in a topological code, see Ref.~\cite{Davydova2024}. 
The generalization extends to qudit systems and nonabelian groups~\cite{Cowtan2022Algebraic}. 
However, for nonabelian groups a definite global charge is not fixed by measuring the charge locally.  
We leave further development of this direction to future work. 
In fact, the representation to be measured need not form the logical operators of a quantum error-correcting code at all, raising the potential of more general fault-tolerant code deformations.

We now address generalizations that address the parallelism and time overhead of the gauging measurement. 
The gauging measurement procedure can be applied directly to an arbitrary number of logical operators in parallel, provided that no pair of these logical operators act on a common qubit via different nontrivial Pauli operators. 
To maintain an LDPC code during the code deformation step it is required that only a constant number of logical operators being measured share support on any single qubit. 
For codes that support many disjoint logical representatives, this offers the potential of performing highly parallelized logical gates. 
It is also possible to trade-off time overhead for space overhead by performing $2m-1$ measurements of equivalent logical operators in parallel for $\frac{d}{m}$ rounds and then taking a majority vote to determine the classical outcome of the logical measurement. 
For codes that do not support many disjoint logical representatives, techniques have been developed to fault-tolerantly attach ancilla systems that increase the number of disjoint logical representatives~\cite{Zhang2024}. 
The gauging measurement procedure can then be applied in parallel to the modified system, see Ref.~\cite{cowtan2025parallel}.

\begin{remark}[Replacing $G$ with a hypergraph]
\label{rem:Hypergraph}
    The gauging measurement procedure can also be generalized to measure a number of operators simultaneously. 
    This procedure can be applied to measure any abelian group of operators that are describable as the $X$-type operators that commute with an auxiliary set of $k$-local $Z$-type checks. 
    This type of group can equivalently be formulated as the kernel of a sparse linear map over $\mathbb{F}_2$ using the stabilizer formalism~\cite{gottesman1997stabilizer} (more generally one has $\mathbb{F}_p$ for $p$ a prime).
    For qubits this is equivalent to replacing $G$ with a hypergraph. 
    The generalized gauging procedure performs a code deformation by introducing a qubit for each hyperedge and measuring into new $A_v$ checks given by the product of $X$ on a vertex and the adjacent hyperedges. 
\end{remark}

We now turn to alternative implementations of the fault-tolerant gauging measurement procedure. 
The scaling of the fault-distance established in Theorem~4 in the main text holds even if the $B_p$ checks are measured much less often than the $A_v$ and $\tilde{s}_i$ checks. 
In fact, they never need to be measured directly as they can be inferred from the initialization and readout steps of the code deformation. 
While this is appealing for cases where the $B_p$ checks have high weight, it results in large detector cells and hence the code is not expected to have a threshold without further modifications. 
Still, this strategy may prove useful in practice for small instances. 

Finally, we point out that the gauging procedure can alternatively be implemented by a circuit by implementing the $A_v$ measurements as follows. 
After initializing the edge qubits we perform the entangling circuit $\prod_v \prod_{e\ni v} CX_{v\rightarrow e}$. 
Next, we projectively measure $X_v$ on all vertices in $G$ and keep the post-measurement state. 
We then repeat the same entangling circuit, followed by a measurement of the edge operators $Z_e$ where the edge qubits are discarded. 
The circuit implementation leads to a different, but closely related, fault-tolerant implementation where the vertex qubits are decoupled and can be discarded during the code deformation. 
The main difference is that this can lead to a reduction in the distance by a constant multiple. 
However, this distance reduction can be avoided by adding an extra dummy vertex to divide each edge into a pair of edges. \newline

\section{Spacetime code and spacetime fault-distance}

In this section we provide further details about the fault-tolerant gauging measurement procedure. 
We follow the general approach to fault tolerance via repeated measurements introduced in Ref.~\cite{dennis2002topological}, and borrow some terminology from Ref.~\cite{McEwen2023}. 
An alternative perspective through the lens of gauge fixing can be found in Refs.~\cite{vuillot2019code} and \cite{cross2024linear}.
The procedure begins at time $t_0$ followed by at least $d$ rounds of syndrome measurements in the original $[[n,k,d]]$ code. 
Next, there is a code deformation step at time $t_i$ followed by at least $d$ rounds of syndrome measurements in the deformed code. 
Finally, there is another code deformation step at time $t_o$ back to the original code, followed by at least $d$ rounds of syndrome measurements in the original code. 

We use a convention for labelling time steps that has check measurements occurring with an offset of half from an integer. 
With this convention detectors and space errors are associated to integer time steps, while measurement errors are associated to integer + half time steps. 
The initial gauging code deformation measurements are made at time step $t_i+\frac{1}{2}$ and the final ungauging code deformation measurements are made at time step $t_o+\frac{1}{2}$. 

\begin{lemma}[Spacetime code]\label{lem:spacetime-code}
    The following form a generating set of the local detectors in the fault-tolerant gauging measurement procedure:\\
    For $t<t_i$ and $t>t_o$
    \begin{itemize}
        \item $s_j^{t}$ which is given by the collection of repeated $s_j$ checks in the original code at times $t-\frac{1}{2},t+\frac{1}{2}$. 
    \end{itemize}
    For $t_i<t<t_o$
    \begin{itemize}
        \item $A_v^{t}$ which is defined by the repeated measurement of the $A_v$ check in the deformed code at times ${t-\frac{1}{2},}{t+\frac{1}{2}}$.
        \item $B_p^t$ which is defined similarly for the $B_p$ check at times $t-\frac{1}{2},t+\frac{1}{2}$. 
        \item $\tilde{s}_j^{t}$ which is defined similarly for the deformed check $\tilde{s}_j$ at times $t-\frac{1}{2},t+\frac{1}{2}$. 
    \end{itemize}
    For $t=t_i$
    \begin{itemize}
        \item $B_p^t$ which is given by the measurement of $B_p$ at time $t_i+\frac{1}{2}$ together with the initialization of the edge qubits $e\in p$ in the $\ket{0}_e$ state at time $t_i-\frac{1}{2}$. 
        \item $\tilde{s}_j^{t_i}$ which is given by the measurement of ${s}_j$, and the initialization of the edge qubits $e\in \gamma$ in the $\ket{0}_e$ state, at time $t_i-\frac{1}{2}$, and the measurement of $\tilde{s}_j$ at time $t_i+\frac{1}{2}$
    \end{itemize}
    For $t=t_o$
    \begin{itemize}
        \item $B_p^t$ which is given by the measurement of $B_p$ at time $t_o-\frac{1}{2}$ and the measurement of $Z_e$ on the edge qubits $e\in p$ at time $t_o+\frac{1}{2}$
        \item $\tilde{s}_j^{t_o}$ which is given by the measurement of $\tilde{s}_j$ at time $t_o-\frac{1}{2}$, together with the measurement of $Z_e$ on the edge qubits $e\in \gamma$, and $s_j$ at time $t_o+\frac{1}{2}$
    \end{itemize}
\end{lemma}

\begin{proof}
    Away from the initial and final steps of the code deformation, the fault-tolerant gauging measurement procedure consists of repeatedly measuring the same checks. 
    In this setting any local detector must contain a pair of measurements of the same check. 
    This step assumes there are no local relations in the original code, see the remark below. 
    Any detector formed by the measurement of the same check at times $(t,t+\kappa),$ can be decomposed into detectors at time steps $(t,t+1)$, $(t+1,t+2),\dots,$ $(t+\kappa-1,t+\kappa)$.
    During the code deformation there are also detectors that involve the initialization of the edge qubits, and their final read-out. 
    Due to the measurement of $A_v$ terms during the code deformation, any detector involving the initialization of edge qubits in the $Z$ basis must include a collection of edge qubits that corresponds to a product of $B_p$ and $\tilde{s}_j$ checks. 
    A similar statement holds for detectors that involve the read-out of edge qubits in the $Z$ basis.
    Any such detector can be decomposed into one of the detectors involving edges at time $t_i$ or $t_o$ listed above, combined with repeated measurement checks. 
    Local detectors that cross the $t_i$ or $t_o$ deformation steps must involve a check $s_j$ from the original code at some time step before or after the deformation, along with a deformed version $\tilde{s_j}$. 
    Such a detector can be decomposed into repeated measurement detectors, and one of the detectors involving edge qubits introduced above. 
\end{proof}

\begin{remark}[Space-only detectors]
    When stating the above lemma, we have assumed there are no local detectors formed by collections of checks at a single time step. 
    Such detectors occur in codes with local relations, or meta-checks, and are useful for single-shot quantum error correction \cite{campbell2019theory}. 
    The fault-tolerant gauging measurement procedure can be applied to codes with local relations. 
    If such local relations are present, it is simple to modify the above lemma to include the space-only local detectors at each time step.
    However, we have chosen not to focus on such codes as the single-shot property of these code does not help reduce the time overhead of the gauging logical measurement when our current scheme is used.  
\end{remark}

\begin{remark}[Initial and final boundary conditions]
\label{rem:initialfinal}
    Following Ref.~\cite{Beverland2024}, we use the convention that the initial and final round of stabilizer measurements are perfect. 
    This is to facilitate a clean statement of our results and should not be taken literally. 
    The justification for why this convention does not fundamentally change our results is due to the $d$ rounds of error correction in the original code before and after the gauging measurement. 
    This ensures that any error process that involves both the gauging measurement and the initial or final boundary condition must have distance greater that $d$. 
    In practice the gauging measurement is intended to be one component of a larger fault-tolerant quantum computation which determines the appropriate realistic boundary conditions to use. 
\end{remark}

\begin{remark}[Spacetime syndromes]
    The syndrome of a fault is the set of detectors it causes to have result $-1$. 
    We say that the fault violates these detectors. 
    The faults can be organized according to the kind of syndrome they create:
    \\
    For $t < t_i$ and $t > t_o$
    \begin{itemize}
        \item A Pauli $X_v$ (or $Z_v$) operator fault at time $t$ violates the $s_j^t$ detectors for all $s_j$ checks that do not commute with $X_v$ (or $Z_v$). 
        \item An $s_j$-measurement fault at time $t+\frac{1}{2}$ violates the detectors $s_j^t$ and $s_j^{t+1}$.
    \end{itemize}
    For $t_i< t < t_o$ 
    \begin{itemize}
        \item A Pauli $X_v$ operator fault at time $t$ violates the $\tilde{s}_j^t$ detectors for all $\tilde{s}_j$ checks that do not commute with $X_v$.
        \item A Pauli $Z_v$ operator fault at time $t$ violates the $A_v^t$ detector and the $\tilde{s}_j^t$ detectors for all $\tilde{s}_j$ checks that do not commute with $Z_v$.
        \item A Pauli $X_e$ operator fault at time $t$ violates the $B_p^t$ detectors for all $p \ni e$ and the $\tilde{s}_j^t$ detectors for all $\tilde{s}_j$ checks that anticommute with $X_e$
        \item A Pauli $Z_e$ operator fault at time $t$ violates the $A_v$ detectors for $v \in e$. 
        \item A $\tilde{s}_j$-measurement fault at time $t+\frac{1}{2}$ violates the detectors $\tilde{s}_j^t$ and $\tilde{s}_j^{t+1}$.
        \item An $A_v$-measurement fault at time $t+\frac{1}{2}$ violates the detectors $A_v^t$ and $A_v^{t+1}$. 
        \item A $B_p$-measurement fault at time $t+\frac{1}{2}$ violates the detectors $B_p^t$ and $B_p^{t+1}$
    \end{itemize}
    For $t=t_i$
    \begin{itemize}
        \item A Pauli $X_v$ (or $Z_v$) operator fault at time $t$ violates the $\tilde{s}_j^{t}$ detectors for all $\tilde{s}_j$ checks that do not commute with $X_v$ (or $Z_v$)
        \item A Pauli $X_e$ operator fault at time $t$ violates the $B_p^t$ detectors for all $p \ni e$ and the $\tilde{s}_j^{t}$ detectors for all $\tilde{s}_j$ checks that anticommute with $X_e$
        \item A $\ket{0}_e$ initialization fault at time $t-\frac{1}{2}$ is equivalent to a Pauli $X_e$ operator fault at time $t$ and so violates the same detectors. 
        \item A $\tilde{s}_j$-measurement fault at time $t+\frac{1}{2}$ violates the detectors $\tilde{s}_j^t$ and $\tilde{s}_j^{t+1}$.
        \item An $A_v$-measurement fault at time $t+\frac{1}{2}$ violates the detector $A_v^{t+1}$. 
        \item A $B_p$-measurement fault at time $t+\frac{1}{2}$ violates the detectors $B_p^t$ and $B_p^{t+1}$
    \end{itemize}
    For $t=t_o$
    \begin{itemize}
        \item An $X_v$ (or $Z_v$) Pauli operator fault at time $t$ violates the $\tilde{s}_j^{t}$ detectors for all $\tilde{s}_j$ checks that do not commute with $X_v$ (or $Z_v$)
         \item A Pauli $X_e$ operator fault at time $t$ violates the $B_p^t$ detectors for all $p \ni e$ and the $\tilde{s}_j^{t}$ detectors for all $\tilde{s}_j$ checks that anticommute with $X_e$
        \item A $Z_e$-measurement read-out fault at time $t+\frac{1}{2}$ is equivalent to a Pauli $X_e$ fault at time $t$ and so violates the same detectors.
        \item An $s_j$-measurement fault at time $t+\frac{1}{2}$ violates the detectors $\tilde{s}_j^t$ and $s_j^{t+1}$.
        \item An $A_v$-measurement fault at time $t-\frac{1}{2}$ violates the detector $A_v^{t-1}$.
    \end{itemize}
\end{remark}

\begin{remark}[Syndrome mobility]
    For $t < t_i$ and $t > t_o$ the syndromes can be created and moved around the code by Pauli errors, and propagated forwards or backwards in time via measurement errors,  as usual. 
    For $t_i < t < t_o$ Pauli errors on vertex qubits behave similarly, with the exception that $Z_v$ errors cause additional $A_v$ syndromes. 
    Pauli $Z_e$ errors on edge qubits form strings that move the $A_v$ syndromes along edge-paths in the graph $G$. 
    Pauli $X_e$ errors on edge qubits produce $B_p$ syndromes and can also produce clusters of $\tilde{s}_j$ syndromes that cannot be generated by Pauli errors on the vertex qubits alone. 
    Again measurement errors can propagate syndromes forwards and backwards in time.
    At the gauging and ungauging time steps $t=t_i$ and $t=t_o$, respectively, $A_v$ syndromes can \textit{condense} that is be created or destroyed at the time slices where the $A_v$ stabilizer measurements start or end. 
    On the other hand, $B_p$ and $\tilde{s}_j$ errors can propagate through the gauging and ungauging time steps by mapping into an error caused by Pauli operators on the vertices alone up to multiplication with spacetime stabilizers including $A_v$ operators, see Lemma~\ref{lem:SpacetimeStabilizers}. 
\end{remark}

\begin{definition}[Spacetime logical fault]
    A spacetime logical fault is a collection of space and time faults that does not violate any detectors. 
\end{definition}

\begin{definition}[Spacetime stabilizer]
    A spacetime stabilizer is a trivial spacetime logical fault in the sense that it is a collection of space and time faults that does not violate any detectors \textit{and} does not affect the result of the gauging measurement procedure.
\end{definition}

\begin{lemma}
    \label{lem:SpacetimeStabilizers}
    The following form a generating set of local spacetime stabilizers:
    \\
    For $t<t_i$ and $t>t_o$
    \begin{itemize}
        \item A stabilizer check operator $s_j$ at time $t$. 
        \item A pair of Pauli $X_i$ (or $Z_i$) faults at times $t,t+1,$ together with measurement faults on all checks $s_j$ that do not commute with $X_i$ (or $Z_i$) at time $t+\frac{1}{2}$.
    \end{itemize}
    For $t_i<t<t_o$
    \begin{itemize}
        \item A stabilizer check operator $\tilde{s}_j$, $A_v$, or $B_p$, at time $t$.
        \item A pair of vertex Pauli $X_v$ faults at times $t,t+1,$ together with measurement faults on all checks $\tilde{s}_j$ that do not commute with $X_v$ at time $t+\frac{1}{2}$.
        \item A pair of vertex Pauli $Z_v$ faults at times $t,t+1,$ together with measurement faults on $A_v$ and all checks $\tilde{s}_j$ that do not commute with $Z_v$ at time $t+\frac{1}{2}$.
        \item A pair of edge Pauli $X_e$ faults at times $t,t+1,$ together with measurement faults on checks $B_p$ with $p\ni e$ and all checks $\tilde{s}_j$that do not commute with $X_e$ at time $t+\frac{1}{2}$.
        \item A pair of edge Pauli $Z_e$ faults at times $t,t+1,$ together with measurement faults on checks $A_v$ with $v\in e$ at time $t+\frac{1}{2}$.
    \end{itemize}
    For $t=t_i$
    \begin{itemize}
        \item A stabilizer check operator $s_j$ or $Z_e$ at time $t$.
        \item A pair of vertex Pauli $X_v$ faults at times $t,t+1,$ together with measurement faults on all checks $\tilde{s}_j$ that do not commute with $X_v$ at time $t+\frac{1}{2}$.
        \item A pair of vertex Pauli $Z_v$ faults at times $t,t+1,$ together with measurement faults on $A_v$ and all checks $\tilde{s}_j$ that do not commute with $Z_v$ at time $t+\frac{1}{2}$.
        \item A pair of edge Pauli $X_e$ faults at times $t,t+1,$ together with measurement faults on checks $B_p$ with $p\ni e$ and all checks $\tilde{s}_j$that do not commute with $X_e$ at time $t+\frac{1}{2}$.
        \item A $\ket{0}_e$ initialization fault at time $t-\frac{1}{2}$ together with a Pauli $X_e$ fault at time $t$. 
        \item A Pauli $Z_e$ edge fault at time $t+1$ together with a pair of $A_v$ measurement faults for $v\in e$ at time $t+\frac{1}{2}$. 
    \end{itemize}
    For $t=t_o$
    \begin{itemize}
        \item A stabilizer check operator $\tilde{s}_j$, $A_v$, or $B_p$ at time $t$.
        \item A pair of vertex Pauli $X_v$ faults at times $t,t+1,$ together with measurement faults on all checks $s_j$ that do not commute with $X_v$ at time $t+\frac{1}{2}$.
        \item A pair of vertex Pauli $Z_v$ faults at times $t,t+1,$ together with measurement faults on all checks $s_j$ that do not commute with $Z_v$ at time $t+\frac{1}{2}$.
        \item A Pauli $X_e$ edge fault at time $t,$ together with a measurement fault on check $Z_e$ at time $t+\frac{1}{2}$.
        \item A Pauli $Z_e$ edge fault at time $t$. 
        \item A Pauli $Z_e$ edge fault at time $t-1$ together with a pair of $A_v$ measurement faults for $v\in e$ at time $t-\frac{1}{2}$.
    \end{itemize}
\end{lemma}

\begin{proof}
    Any nonempty local spacetime stabilizer must involve a Pauli operator, or equivalent initialization or read-out error, as otherwise the stabilizer would have to include measurement errors on all repeated measurements of some check. 
    If a nontrivial local spacetime stabilizer contains a Pauli operator at some time, it must be a space stabilizer or contain a Pauli operator at another time, or an equivalent state intialization or measurement read-out error. 
    The product of the Pauli operators from all time steps involved must itself be a space stabilizer, where we are treating initialization and read-out errors as equivalent to some Pauli operator error. 
    Any local spacetime fault of this form can be constructed from a product of the spacetime stabilizers introduced above by first reconstructing the operators at the earliest time step at the cost of creating matching operators at the next time step, and so on until the final time step, where the product of the operators must now become trivial. 
    This leaves a local spacetime stabilizer with only measurement errors, which must also be trivial.
    Hence, the original fault pattern is a product of the introduced spacetime stabilizers as claimed.
\end{proof}

\begin{definition}[Spacetime fault-distance]
    The spacetime fault-distance is the weight, counted in terms of single-site Pauli errors and single measurement errors, of the minimal collection of faults that does not violate any detectors and is not a spacetime stabilizer.
\end{definition}

\begin{lemma}[Time fault-distance]\label{lem:time-fault-distance}
    The fault-distance for measurement and initialization errors is $(t_o-t_i)$. 
\end{lemma}

\begin{proof}
    With the convention that includes one round of perfect measurements at the initial and final step of the whole procedure, see Remark~\ref{rem:initialfinal}, all pure measurement logical faults must start and end on the code deformation steps. 
    Otherwise a measurement fault at time step $t$ must be followed and preceded by another measurement fault on the same type of check at time steps $t-1,t+1$. 
    At this point it follows that a measurement and initialization logical fault must have distance at least $(t_o-t_i)$ because that is the number of measurement rounds between $t_i$ and $t_o$.
    
    We now proceed to discuss the measurement logical faults more explicitly. 
    At the initial code deformation step logical measurement faults can terminate in two ways. 
    First, a string of $A_v$ measurement faults can terminate since a fault on the initial measurement of $A_v$ only violates the $A_v^{t_i}$ detector. 
    Second, a collection of measurement errors on the set of $B_p$ and $\tilde{s}_j$ checks that anticommute with some $X_e$ operator can terminate since a fault on a $\ket{0}_e$ edge initialization violates the $B_p^{t_i}$ and $\tilde{s}_j^{t_i}$ detectors for all $B_p$ and $\tilde{s}_j$ checks that anticommute with $X_e$. 
    Similarly, at the final code deformation step a string of $A_v$ measurement errors can terminate and an appropriate collection of $B_p$ and $\tilde{s}_j$ measurement errors can terminate. 
    From this, we see that all measurement and initialization logical faults are generated by repeated measurement errors on either a check $A_v$, or an appropriate collection of $B_p$ and $\tilde{s}_j$ checks, at all time steps between $t_i$ and $t_o$.  
    The logical fault given by repeated measurement error on a collection of $B_p$ and $\tilde{s}_j$ checks is in fact a trivial logical fault as it can be decomposed into a product of spacetime stabilizers that consist of an initialization or $X_e$ error at some time step, followed by a read-out or $X_e$ error at the next time step, and a collection of measurement errors on all $B_p$ and $\tilde{s}_j$ checks that anticommute with $X_e$ between the time steps. 
    On the other hand, a fault on all repeated measurements of an $A_v$ check results in a logical error as it changes the inferred value of the logical measurement. 
    Hence, the lower bound $(t_o-t_i)$ on the measurement fault-distance is saturated. 
\end{proof}

If we do not assume a round of perfect stabilizer check measurements at the start and end of the procedure, it is possible that there are additional logical measurement faults that extend from the initial and final step of the code deformation to the start and end of the whole procedure. 
For this reason we have included $d$ rounds of repeated stabilizer measurements in the undeformed code before and after the gauging measurement code deformation.

\begin{lemma}[Decoupling of space and time faults]
    \label{lem:spacetimedecoupling}
    Any spacetime logical fault is equivalent to the product of a space logical fault and a time logical fault, up to multiplication with spacetime stabilizers.
\end{lemma}

\begin{proof}
    Consider an arbitrary spacetime logical fault $F$. 
    The space component of $F$, consisting of Pauli operators, can be cleaned into any single timestep $t_i \leq t \leq t_o$ via multiplication with spacetime stabilizers that involve like Pauli operator faults at time steps $t, t+1$ together with appropriate measurement faults. 
    In particular, we can clean the space component of the logical fault to time step $t_i$. 
    In the cleaned spacetime logical fault all measurement errors must occur in the time steps $t_i\leq t \leq t_o$ as measurement errors outside this time window must propagate to the initial or final time boundary, which has no measurement errors by assumption. 
    The measurement faults form strings that propagate through time from $t_i\pm \frac{1}{2}$ to $t_o\pm \frac{1}{2}$. 
    These strings must end either at time step $t_o-\frac{1}{2}$ on an $A_v$ measurement fault, or at $t_o+\frac{1}{2}$ on a $Z_e$ measurement fault.
    The strings ending on $A_v$ measurement faults are timelike logical faults. 
    The strings ending on a $Z_e$ measurement fault can all be assumed to originate from a corresponding $\ket{0}_e$ initialization fault at time step $t_i - \frac{1}{2}$ by multiplying with spacetime stabilizers that introduce pairs of $\ket{0}_e$ and $X_e$ faults. 
    After this equivalence, the strings ending on $Z_e$ measurement faults are also timelike logical faults. 
    Hence, up to spacetime stabilizer equivalence the original logical can be deformed into a product of timelike logical faults and a residual spacelike fault which must also be a logical fault due to linearity.
\end{proof}

\begin{theorem}[Spacetime fault-distance]\label{thm:spacetime-fault-distance}
    The spacetime fault-distance of the fault-tolerant gauging measurement procedure is $d$. Here we are assuming that a sufficiently expanding graph, satisfying $h(G)\geq 1$, and a sufficient number of rounds of repeated measurements, satisfying $(t_o-t_i)\geq d$, are used in the procedure. 
\end{theorem}

\begin{proof}
    First, consider a spacetime logical fault that is not equivalent to a spacelike logical fault. 
    Such logical faults must have support on all time steps $t_i<t<t_o$ and hence their distance is lower bounded by $d$, assuming $(t_o-t_i)\geq d$. 
    Now consider a spacetime logical fault that is equivalent to a spacelike logical fault. 
    Following the proof of Lemma~\ref{lem:spacetimedecoupling}, such a fault can be deformed into a spacelike logical fault at time step $t_i$ via multiplication with spacetime stabilizers that involve like Pauli operators at time steps $t,t+1$, and spacetime stabilizers at time step $t_i$ that introduce $\ket{0}_e$ initialization faults along with an $X_e$ Pauli faults. 
    Following Lemma~2 in the main text, the space distance of the resulting spacelike logical is lower bounded by $d$ assuming the Cheeger constant of $G$ satisfies $h(G)\geq 1$. 
    Undoing the spacetime stabilizer equivalence that cleaned the spacetime logical into a spacelike logical cannot reduce the distance, since the combined weight of Pauli and initialization faults cannot be reduced below that of the spacelike logical by the spacetime stabilizers that were used in the cleaning process. 
    This is because each such stabilizer preserves the parity of the space and initialization faults along the timeline ata fixed position. 
\end{proof}

Our proof of the spacetime fault-distance applies equally well if the plaquette checks $B_p$ are high weight, and if they are measured less frequently than every time step. 
In fact, it even holds if the $B_p$ detectors are only inferred once, via initialization and final readout, avoiding the need to measure high weight operators. 
However, in this case the procedure is likely not scalable in the sense that it likely does not have a threshold against uncorrelated random Pauli and measurement noise on all fault sites. 
Applications of this procedure to small fixed instances remains an interesting direction.

\end{document}